\documentclass[twoside,journal]{IEEEtran}
\usepackage{amsfonts}
\usepackage{amssymb}
\usepackage{amsmath}
\usepackage{mathrsfs}
\usepackage{verbatim}
\usepackage{subfigure}
\usepackage{balance}
\usepackage{booktabs}
\usepackage{fancybox}
\usepackage{bm}
\usepackage{changepage}
\usepackage{extarrows}
\usepackage{algorithm}
\usepackage{algorithmic}
\usepackage{diagbox}
\usepackage{multirow}
\usepackage{array}
\usepackage{graphicx}
\usepackage{epstopdf}
\usepackage{caption}
\usepackage{color}
\newtheorem{theorem}{Theorem}
\newtheorem{lemma}{Lemma}
\newtheorem{corollary}{Corollary}
\newtheorem{proof}{Proof}
\newtheorem{proposition}{Proposition}

\newtheorem{remark}{Remark}

\begin{document}

\title{Physical-Layer Security in the Finite Blocklength Regime over Fading Channels}
\author{Tong-Xing~Zheng,~\IEEEmembership{Member,~IEEE,}
Hui-Ming~Wang,~\IEEEmembership{Senior~Member,~IEEE,}\\
Derrick~Wing~Kwan~Ng,~\IEEEmembership{Senior~Member,~IEEE,}
and~Jinhong~Yuan,~\IEEEmembership{Fellow,~IEEE}
\thanks{T.-X.~Zheng is with the School of Information and Communications Engineering, Xi'an Jiaotong University, Xi'an 710049, China, also with the State Key Laboratory of Integrated Services Networks, Xidian University, Xi'an 710071, China, and also with the Ministry of Education Key Laboratory for Intelligent Networks and Network Security, Xi'an Jiaotong University, Xi'an 710049, China 
(e-mail: zhengtx@mail.xjtu.edu.cn).}
\thanks{H.-M. Wang is with the School of Information and Communications Engineering, Xi'an Jiaotong University, Xi'an 710049, China, and also with the Ministry of Education Key Laboratory for Intelligent Networks and Network Security, Xi'an Jiaotong University, Xi'an 710049, China 
(e-mail: xjbswhm@gmail.com).}
\thanks{D. W. K. Ng and J. Yuan are with the School of Electrical Engineering and Telecommunications, University of New South Wales, Sydney, NSW 2052, Australia (e-mail: w.k.ng@unsw.edu.au; j.yuan@unsw.edu.au).}
}
\maketitle
\vspace{-0.8 cm}
	
\begin{abstract}
This paper studies physical-layer secure transmissions from a transmitter to a legitimate receiver against an eavesdropper over slow fading channels, taking into account the impact of  finite blocklength secrecy coding. 
A comprehensive analysis and optimization framework is established to investigate secrecy throughput for both single- and multi-antenna transmitter scenarios.
Both adaptive and non-adaptive design schemes are devised, in which the secrecy throughput is maximized by exploiting the instantaneous and statistical channel state information of the legitimate receiver, respectively. 
Specifically, optimal transmission policy,  blocklength, and code rates are jointly designed to maximize the secrecy throughput. 
Additionally, null-space artificial noise is employed to improve the secrecy throughput for the multi-antenna setup with the optimal power allocation derived.
{Various important insights are developed. In particular, 1) increasing blocklength benefits both reliability and secrecy under the proposed transmission policy; 2) secrecy throughput monotonically increases with blocklength; 3) secrecy throughput initially increases but then decreases as secrecy rate increases, and the optimal secrecy rate maximizing the secrecy throughput should be carefully chosen in order to strike a good balance between rate and decoding correctness. }
Numerical results are eventually presented to verify theoretical findings.
\end{abstract}

	
\IEEEpeerreviewmaketitle
	
\section{Introduction}
{In the past decade, pursuing communication security at the physical layer has received a considerable interest, e.g.,}
\cite{Bloch2011Physical}-\cite{Wu2018Survey}. 
{In particular}, physical-layer security exploits the inherent randomness of noise and wireless channels to protect wireless secure transmissions {\cite{Deng2016Artificial}-\cite{Yan2018Three}}, which can provide an additional {mechanism for security guarantee and} can coexist with those security techniques already employed at the upper layers, such as key-based encipherment. Most recent progress in developing physical-layer security is motivated by Wyner's pioneering work. {Specifically, }the concept of secrecy capacity {was first established which is defined as} the supremum of secrecy rates at which both reliability and secrecy are achieved over a wiretap channel \cite{Wyner1975Wire-tap}. Wyner showed that the error probability and information leakage can be made arbitrarily low {concurrently} with {an} appropriate secrecy coding, provided that a data rate below the secrecy capacity is chosen and meanwhile the data is mapped to asymptotically long codewords, i.e., the coding  blocklength tends to infinity.
However, the upcoming 5G wireless communication system{s are required} to support various novel traffic types 
{adopting} short packets to reduce the {end-to-end communication} latency, e.g., smart{-}traffic safety and machine-to-machine communications {\cite{Durisi2016Toward,Wong2017Key}}. For the short-packet applications, conventional physical-layer security schemes originated from infinite blocklength {are generally} suboptimal and the {impact of} finite blocklength could be destructive for secure communications. {Therefore,} it is necessary to rethink the analysis and design of physical-layer security for the finite blocklength regime.

\subsection{Previous Works and Motivations}
Decoding {with} finite blocklength will inevitably reduce the secrecy capacity and some {preliminary works} have been devoted to analyzing the impact of finite blocklength on secrecy for the wiretap channel.
For example, the authors in \cite{Cao2015Physical} derived an upper bound for the information leakage probability for a {given} target decoding error probability {demonstrating} the inherent trade-off between secrecy and reliability.
The authors in \cite{Yang2016Finite} provided both upper and lower bounds for the maximal secrecy rate {capturing the impact of  finite blocklength, error probability, and information leakage in both} degraded discrete-memoryless wiretap channels and Gaussian wiretap channels. The obtained bounds were shown to be tighter than existing ones from \cite{Yassaee2013Non,Tan2012Achievable}. 
The work in \cite{Yang2016Finite} was further extended by \cite{Yang2017Secrecy}, in which the optimal second-order secrecy rate was derived for a semi-deterministic wiretap channel, and the optimal tradeoff between secrecy and reliability {with} finite blocklength was analytically characterized. It should be noted that, all the above works were {aimed} to uncover the fundamental limits of secrecy performance from the {information theory point of view}, whereas {the design of practical signaling and transmission schemes were not investigated}.

{In practice, due to finite blocklength penalty for practical coding schemes}, even a secrecy rate below the secrecy capacity cannot guarantee a perfectly successful {and secure} communication. In this sense, in addition to exploring and/or improving the fundamental limits of the maximal secrecy rate, optimizing secrecy throughput seems more important from the perspective of transmission efficiency, particularly for fading channels where the code rates can be {adapted} to the fading status. Herein, the secrecy throughput {denotes} the amount of successfully delivered secret information subject to certain reliability and secrecy constraints.
{In fact}, the secrecy throughput has extensively been taken as {an} optimization objective for the {design of} secure transmissions in slow fading channels in the context of infinite blocklength  \cite{Zhang2013Design}-\cite{Zheng2018Secure}. {Nevertheless, to optimize the secrecy throughput under the constraint of finite blocklength is difficult, and the results derived for infinite blocklength, e.g., \cite{Zhang2013Design}-\cite{Zheng2018Secure}, cannot be directly applied. Indeed, the blocklength itself is an optimization variable, and it couples with other variables in a sophisticated manner which makes the optimization problem intractable.} 
For instance, the authors in a recent work \cite{Farhat2018Secure} investigated the secrecy throughput of a relay-aided secure transmission with finite blocklength, where neither the instantaneous channel state information (CSI) with respect to (w.r.t.) the legitimate receiver nor the eavesdropper is available at the transmitter side. Numerical results were presented therein to show that there exists a critical value of the blocklength that maximizes the secrecy throughput.

Despite the above endeavors, there are  {some} fundamental questions regarding the design of physical-layer security schemes  {with finite blocklength} that have not been thoroughly addressed. First of all, a theoretical proof of the optimal blocklength  {and the corresponding} secrecy rate for maximizing the secrecy throughput is of great significance for the practical design of secure transmissions, which however has not yet been reported by existing literature.
Also, in many applications, the transmitter is capable to acquire the instantaneous CSI of the legitimate receiver {in slow fading channels} via training or feedback. {Yet,} the potential of {exploiting} the instantaneous CSI {to alleviate the negative impact} of finite blocklength {on the performance of secure communications} has not been exploited. Furthermore, only the single-antenna transmitter scenario has been {considered, e.g., \cite{Cao2015Physical}-\cite{Yang2017Secrecy}, \cite{Farhat2018Secure}}, and the design of the optimal signaling and code rates for multi-antenna systems with finite blocklength is still an open issue. This research work aims to {provide an analytical framework and design schemes to address the abovementioned problems}.

\subsection{Contributions}
This paper investigates the security issue between a pair of legitimate communicating parties in the presence of an eavesdropper, considering {the impact of finite blocklength in secrecy coding}. The secrecy throughput is thoroughly analyzed and optimized for  both single- and multi-antenna transmitter scenarios. In particular, both adaptive and non-adaptive parameter design schemes are proposed for  each scenario. The main contributions of this work are summarized as follows:
\begin{itemize}
	\item 
	For the single-antenna transmitter scenario, the secrecy throughput is {maximized} by jointly {optimizing} the transmission policy, blocklength, as well as  code rates. {Closed-form bounds and approximations for the secrecy rate are provided to facilitate the practical design of code rates for achieving a close-to-optimal performance.}
	\item
	For the multi-antenna transmitter configuration, the optimality of the null-space artificial noise (AN) scheme in terms of secrecy throughput maximization is first {investigated}.
	Afterwards, the optimal transmission policy,  blocklength, code rates, and power allocation between the information-bearing signal and the AN are derived. Particularly, the power allocation and the secrecy rate are designed via the alternating optimization method, and {their impacts on the system performance} are further revealed.
	\item
{ Numerous useful insights into the design of secure transmissions are provided with finite blocklength.
	For example, 1) increasing the blocklength can improve both reliability and secrecy, with properly exploiting the instantaneous CSI of the main channel and the statistical CSI of the wiretap channel, which has not been revealed by existing literature, e.g., \cite{Cao2015Physical}-\cite{Yang2017Secrecy};
	2) using the maximal blocklength is profitable for boosting the secrecy throughput, which is distinguished from the observation in \cite{Farhat2018Secure};
	3) due to the finite blocklength penalty, there is a critical secrecy rate that can maximize the secrecy throughput even for the adaptive scheme, rather than always employing the maximal available secrecy rate, which is fundamentally different from the phenomenon with infinite blocklength, e.g., \cite{Zhang2013Design,Zheng2015Multi}.}
\end{itemize}

\subsection{Organization and Notations}
The remainder of this paper is organized as below. 
Section II describes the system model and the underlying optimization problem. 
Sections III and IV detail the secrecy throughput maximization for both single- and multi-antenna transmitter scenarios.  
Section V draws a conclusion.

\emph{Notations}: Bold lowercase letters denote column vectors. $|\cdot|$, $\|\cdot\|$, $(\cdot)^{\dagger}$, $(\cdot)^{\rm T}$, $\ln(\cdot)$, $\mathbb{P}\{\cdot\}$, $\mathbb{E}_v[\cdot]$ denote the absolute value, Euclidean norm, conjugate, transpose, natural logarithm, probability, and the expectation over a random variable $v$, respectively.
$f_v(\cdot)$ and $\mathcal{F}_v(\cdot)$ denote the probability density function (PDF) and cumulative distribution function (CDF) of $v$, respectively. $F^{-1}(\cdot)$ denotes the inverse function of a function $F(\cdot)$. $\mathcal{CN}(\mu,\sigma^2)$ denotes the circularly symmetric complex Gaussian distribution with mean $\mu$ and variance $\sigma^2$.

\section{System Model and Problem Description}

\subsection{Channel Model}
\begin{figure}[!t]
	\centering
	\includegraphics[width = 3.4in]{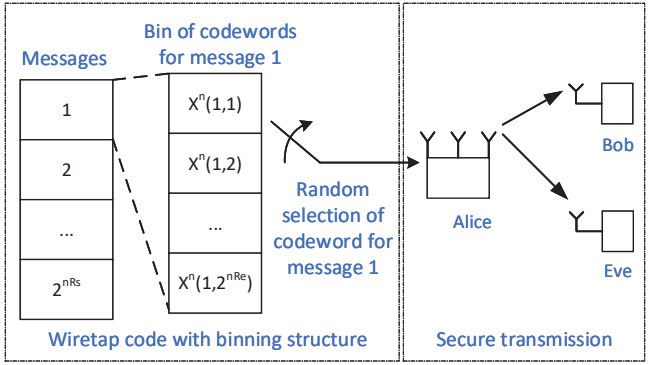}
	\caption{Secure transmission from Alice (multi-antenna) to Bob (single-antenna) overheard by Eve (single-antenna). Alice adopts a wiretap code with a binning structure, where $2^{nR_s}$ messages each are mapped to a bin of $2^{nR_e}$ codewords with a finite blocklength $n$. A codeword among a set of codewords representing the same message is randomly chosen for transmission \cite{Wyner1975Wire-tap}.} 
	\label{SYSTEM_MODEL}
\end{figure}

Consider a secure transmission from a transmitter (Alice) to a legitimate receiver (Bob) coexisting with an eavesdropper (Eve), as depicted in Fig. \ref{SYSTEM_MODEL}.
 Alice is equipped with $M\ge 1$ transmit antennas, whereas Bob and Eve are single-antenna devices. { Quasi-static Rayleigh fading channels are considered, where the channel coherence time is on the order of the blocklength. More specifically, the fading coefficients are assumed to remain constant during the transmission of an entire codeword, but change independently and randomly between two codewords \cite{Bloch2011Physical}.} Denote the coefficients of the main and wiretap channels by $\bm h_b$ and $\bm h_e$, and each entry of $\bm h_b$ and $\bm h_e$ follow the Gaussian distribution $\mathcal{CN}(0,\sigma_b^2)$ and $\mathcal{CN}(0,\sigma_e^2)$, respectively.\footnote{The subscripts $b$ and $e$ are used to refer to Bob and Eve, respectively.} {A common hypothesis is adopted \cite{Zhang2013Design,Zheng2015Multi}, i.e., Bob and Eve know  perfectly the instantaneous CSI of their individual channels $\bm h_b$ and $\bm h_e$, and Alice has the instantaneous CSI of Bob's channel $\bm h_b$ but does not has the instantaneous CSI of Eve's channel $\bm h_e$. Besides, the statistics of both channels $\bm h_b$ and $\bm h_e$ are available at Alice.} Assume that $\bm h_b$, $\bm h_e$, and the receiver noise are mutually independent, where noise variances at Bob and Eve are denoted by $w_{b}^2$ and $w_{e}^2$, respectively.
Alice adopts a constant transmit power $P$. {For notational simplicity, define} $P_b\triangleq \frac{P}{w_b^2}$ and $P_e\triangleq \frac{P}{w_e^2}$ as the normalized power for Bob and Eve, respectively.

\subsection{Finite Blocklength Secrecy Coding}
{ To safeguard information confidentiality, secrecy coding should be employed to encode the secret information bits. Instead of investigating any explicit practical constructions of secrecy codes, the Wyner's wiretap code \cite{Wyner1975Wire-tap}, as a generic code structure, is employed in this paper. A synopsis of the state-of-the-art coding schemes for wiretap channels can be found in \cite{Harrison}.
	
	It is reported in \cite{Bloch2011Physical} that the Wyner's wiretap code possesses a binning structure,} as illustrated in Fig. \ref{SYSTEM_MODEL}, where $2^{nR_s}$ messages are encoded to $2^{nR_t}$ codewords, and each message is mapped to a bin of $2^{nR_e}$ codewords. Here, $n$ denotes the blocklength (i.e., the codeword length or the number of channel uses), $R_s$ and $R_t$ (bits/s/Hz/channel) denote the secrecy rate and codeword rate, respectively. The binning codeword rate, i.e., the rate redundancy $R_e = R_t - R_s$, reflects the cost of providing secrecy.

It is well-known that, for {an} infinite blocklength with $n\rightarrow\infty$, as long as the codeword rate $R_t$ is not larger than Bob's channel capacity, Bob can recover messages {with an arbitrarily low decoding error probability}. On the other hand, perfect secrecy cannot always be guaranteed due to the absence of Eve's instantaneous CSI: once the rate redundancy $R_e$ falls below Eve's channel capacity, perfect secrecy is compromised, and a secrecy outage event is said to have occurred. 
Nevertheless, in the finite blocklength regime which is restricted to a finite number of channel uses, no {practical protocols} can achieve perfectly reliable communications \cite{Bertsekas1992Data}. {Hence, to capture the impact of finite blocklength}, the maximal channel coding rate for sustaining a desired decoding error probability $\epsilon$ at a finite blocklength $n$ (e.g., $n\ge 100$) for a given signal-to-noise ratio (SNR) $\gamma$ {was studied in \cite{Polyanskiy2010Channel} and can be} approximated by  
\begin{equation}
R(\gamma, n,\epsilon) \approx C(\gamma) - \sqrt{\frac{V(\gamma)}{n}}Q^{-1}(\epsilon),
\end{equation}
where $C(\gamma)\triangleq\log_2(1+\gamma)$ denotes the Shannon channel capacity, $V(\gamma)\triangleq \left(1-(1+\gamma)^{-2}\right)\log_2^2 e$ denotes the channel dispersion \cite{Polyanskiy2010Channel}, 
 and $Q(x)$ is the $Q$-function defined as $Q(x)\triangleq \frac{1}{\sqrt{2\pi}}\int_x^{\infty}e^{-\frac{t^2}{2}}dt$.
 Equivalently, the decoding error probability {for a given coding rate $R$} can be expressed as
 \begin{equation}\label{error_pro}
 \epsilon(\gamma, n,R)= Q\left(\frac{C(\gamma)-R}{\sqrt{V(\gamma)/n}}\right).
 \end{equation}
 For ease of notation, let $C_i \triangleq C(\gamma_i)$ and $V_i \triangleq V(\gamma_i)$ for $i\in\{b,e\}$, {where $\gamma_i$ denotes the corresponding SNR}. Define the successful decoding probability of Bob as the complement of its decoding error probability with the codeword rate $R_t$. Then, the successful decoding probability conditioned on the power gain of the main channel, i.e., $\eta\triangleq \|\bm h_b\|^2$, can be expressed as
 \begin{equation}\label{pe}
 p_s(\eta) \triangleq 1 - \epsilon(\gamma_b,n,R_t)
 =1-Q\left(\frac{C_b-R_t}{\sqrt{V_b/n}}\right).
 \end{equation}
The secrecy performance is characterized by the information leakage probability defined below:
\begin{equation}\label{oe}
\mathcal{O}_e \triangleq \mathbb{E}_{\gamma_e}\left[1-\epsilon(\gamma_e,n,R_e)\right].
\end{equation}
\begin{remark}
	Due to the finite blocklength, 
	the secrecy metric information leakage probability in \eqref{oe} appears to be distinguished from the widely used secrecy outage probability, defined as $\mathbb{P}\{R_e\le C_e\}$ \cite{Zheng2015Multi}, for the infinite blocklength regime with $n\rightarrow\infty$.
\end{remark}

\subsection{Optimization Problem}
{ Since Alice knows Bob's instantaneous CSI perfectly, she is able to adapt the code rates to the instantaneous channel gain $\eta$, which implies that the code rates can be functions of $\eta$. 
This paper focuses on the metric named secrecy throughput (bits/s/Hz/channel), which measures the average successfully transmitted information bits per second per Hertz per channel use subject to a secrecy constraint $\mathcal{O}_e\le\delta$, where $\delta\in[0,1]$ is a pre-established threshold for the information leakage probability. Formally, the secrecy throughput is defined as 
 \begin{equation}\label{def_st}
 \mathcal{T}\triangleq\mathbb{E}_{\eta}\left[R_s(\eta)p_s(\eta)\right]~{\rm s.t.}~\mathcal{O}_e\le\delta,
 \end{equation}
which is averaged over $\eta$. 
Note that the introduction of finite blocklength leads to a different definition of secrecy throughput compared to the case of infinite blocklength which is $\mathcal{T}\triangleq\mathbb{E}_{\eta}\left[R_s(\eta)\right]$  \cite{Zhang2013Design,Zheng2015Multi}.
In addition, as will be shown later, in order to meet certain secrecy and reliability requirements during the transmission period, an on-off transmission policy is required;\footnote{The on-off policy was initially proposed for  ergodic-fading channels \cite{Gopala2008Secrecy}, where a codeword experiences many channel realizations. It was later introduced to slow fading channels and well characterized the condition for secure   transmissions  \cite{Zhang2013Design}.} i.e., the transmission should take place only when the channel gain $\eta$ exceeds some threshold $\mu>0$. 
With the on-off policy, $R_s(\eta)$ is set to zero for $\eta<\mu$. }

This paper aims to maximize the secrecy throughput by designing the optimal on-off threshold, signaling, blocklength, as well as code rates.
The following two sections will detail the optimization for single- and multi-antenna transmitter scenarios, respectively. For each scenario, both adaptive and non-adaptive design schemes are examined, where Alice adjusts the arguments based on the instantaneous and statistical CSI of the main channel, respectively.

\section{Single-Antenna Transmitter Scenario}
For the single-antenna transmitter scenario, the SNRs of Bob and Eve are given by $\gamma_b = P_b\eta$ with $\eta = |h_b|^2$ and $\gamma_e = P_e|h_e|^2$, respectively. Clearly, $\gamma_i$ is exponentially distributed with mean $\Gamma_i = P_i \sigma_i^2$ for $i\in\{b,e\}$.
The subsequent two subsections {aim} to maximize the secrecy throughput $\mathcal{T}$ defined in \eqref{def_st} by jointly designing the on-off threshold $\mu(\eta)$, the wiretap code rates $R_s(\eta)$ and $R_e(\eta)$, and the blocklength $n(\eta)$, via adaptive and non-adaptive ways, respectively. For notational convenience, these parameters are treated as functions of $\eta$ by default for the adaptive scheme, with the notation $\eta$ being dropped, and  $\mathcal{T}_{\rm A}$ and  $\mathcal{T}_{\rm N}$ are used to differentiate the adaptive scheme to its non-adaptive counterpart. The optimization problem then can be formulated as below:
\begin{subequations}\label{st_max_problem}
	\begin{align}\label{st_max}
	\max_{\mu>0, R_e>0, R_s>0, n} ~&\mathcal{T}=\mathbb{E}_{\eta}\left[R_sp_s\right]\\
	{\rm s.t.} 
	~~&C_b\ge R_t =  R_s+R_e,~\forall\eta>\mu ,\label{st_max_c1}\\
	~&\mathcal{O}_e\le\delta,\label{st_max_c2} \\
	\label{st_max_c3}
	~&1\le n \le N,~n,N\in\mathbb{Z}^{+}.
	\end{align}
\end{subequations}
Note that \eqref{st_max_c1} is interpreted as a reliability requirement since otherwise the successful decoding probability $p_s$ in \eqref{pe} falls below $0.5$ and it is no better than random guessing, which is definitely not acceptable;  \eqref{st_max_c2} describes the secrecy constraint; \eqref{st_max_c3} is related to a latency constraint, where the integer $N$ denotes the maximal available blocklength imposed by a maximal tolerable delay. 

\subsection{Adaptive Optimization Scheme}\label{ada_single}
In the adaptive scheme, the parameters $\mu$, $R_s$, $R_e$, and $n$ are designed based on $\eta$, i.e., they are adjusted in real time.
A detailed optimization procedure is provided as follows.
\subsubsection{Solving $R_e$}
Since $Q$-function $Q(x)$ is a monotonically decreasing function of $x$, it is known that $p_s$ defined in \eqref{pe} decreases with $R_e$ for a fixed $R_s$. This suggests that, the optimal $R_e$ maximizing $\mathcal{T}_{\rm A}$ should be the minimal $R_e$ that satisfies the secrecy constraint $\mathcal{O}_e\le\delta$. Now that $\mathcal{O}_e$ in \eqref{oe} decreases with $R_e$, the optimal $R_e$ is given as the inverse of $\mathcal{O}_e$ at $\delta$,  i.e.:
\begin{equation}\label{opt_re_ad}
R_e^*= \mathcal{O}_e^{-1}(\delta).
\end{equation}
{ Obviously, $R_e^*$ is independent of $\eta$, but monotonically decreases with $\delta$. This is intuitive that a larger rate redundancy is required to combat the eavesdropper in order to meet a more rigorous secrecy constraint.}  Although it is difficult to derive a closed-form expression for $R_e^*$ due to the complicated $Q$-function, the value of $R_e^*$ can be efficiently acquired via a bisection method with $\mathcal{O}_e(R_e)=\delta$, requiring only the computation {of $Q(x)$ or a lookup table}.

\subsubsection{Solving $\mu$}
The secrecy throughput $\mathcal{T}_{\rm A}$ given in \eqref{st_max} can be calculated as 
 \begin{equation}\label{max_st_ad}
 \mathcal{T}_{\rm A}=\int_{P_b\mu}^{\infty}R_sp_sf_{\gamma_b}(\gamma)d\gamma,
 \end{equation}
where $f_{\gamma_b}(\gamma) = \frac{1}{\Gamma_b}e^{-{\gamma}/{\Gamma_b}}$ is the PDF of $\gamma_b= P_b\eta$.
{It appears} that choosing $\mu$ as small as possible is beneficial for increasing $\mathcal{T}_{\rm A}$, on the premise of satisfying the reliability constraint \eqref{st_max_c1}.
In addition, constraint \eqref{st_max_c1} suggests that $C_b>R_e^*\Rightarrow \eta=\frac{\gamma_b}{P_b}>\frac{2^{R_e^*}-1}{P_b}$ must be ensured to achieve a positive $R_s$.
Hence, the optimal on-off threshold is given by
\begin{equation}\label{opt_mu_ad}
\mu^* = \frac{2^{R_e^*}-1}{P_b}.
\end{equation}
{ This result indicates that the transmission condition for the adaptive scheme is determined by the secrecy constraint. Apparently, $\mu^*$ is monotonically decreasing with $\delta$ since $R_e^*$ decreases with $\delta$. This implies, a weaker channel is still allowed for transmission for a looser secrecy constraint.}

Once $\mu$ is obtained, to maximize $\mathcal{T}_{\rm A}$ in \eqref{max_st_ad} only calls for maximizing $\mathcal{T}_{\rm A}(\eta)\triangleq R_sp_s$ which is conditioned on $\eta$. The subproblem is described as below:
\begin{equation}\label{st_max1}
\max_{R_s, n} ~\mathcal{T}_{\rm A}(\eta)= R_sp_s
~~{\rm s.t.}~~ \eqref{st_max_c3},~0\le R_s\le C_b - R_e^*.
\end{equation}
The basic idea to tackle the above problem is first to maximize $p_s$ over $n$ for a fixed $R_s$ and then to design the optimal $R_s$ that maximizes $R_sp_s$ with the optimal $n$.

\subsubsection{Solving $n$}
For any fixed $R_t\le C_b$, there is no doubt that $p_s$ increases with $n$. However, as shown in \eqref{oe}, {$\epsilon(\gamma_e,n,R_e)$ decreases with $n$ for $R_e\le C_e$ but increases with $n$ otherwise. Then, it remains unclear how
	$\mathcal{O}_e$ defined in \eqref{oe}, as well as $R_e^*$ in \eqref{opt_re_ad}, varies with $n$. More importantly, it is less obvious if the monotonicity of $p_s$ w.r.t. $n$ can still hold, since $R_t = R_s+R_e^*$ becomes independent of $n$. 
Therefore, in order to derive the optimal $n^*$ maximizing $\mathcal{T}_{\rm A}(\eta)$ in \eqref{st_max1}, the monotonicity of $\mathcal{O}_e$ or $R_e^*$ w.r.t. $n$ should be first identified.
\begin{lemma}\label{lemma_opt_re_n}
$\mathcal{O}_e$ in \eqref{oe} and $R_e^*$ in \eqref{opt_re_ad} decrease with $n$.
\end{lemma}
\begin{proof}
	Please refer to Appendix \ref{appendix_lemma_opt_re_n}.
\end{proof}

{ Lemma \ref{lemma_opt_re_n} shows that increasing the blocklength is beneficial for decreasing the information leakage probability such that the required rate redundancy of the wiretap code can be lowered.
This result is perhaps counter-intuitive, which makes sense when one realizes that a larger blocklength will yield a larger decoding error probability for Eve if Eve's channel capacity falls below the rate redundancy.}
With Lemma \ref{lemma_opt_re_n}, the  monotonicity of $\mathcal{T}_{\rm A}(\eta)$ w.r.t. $n$ is uncovered, followed by the optimal $n^*$ that maximizes $\mathcal{T}_{\rm A}(\eta)$.
\begin{theorem}\label{theorem_opt_n}
$\mathcal{T}_{\rm A}(\eta)$ in \eqref{st_max1} increases with $n$ and is maximized at	$n^* = N$.
\end{theorem}
\begin{proof}
	Please refer to Appendix \ref{appendix_theorem_opt_n}.
\end{proof}

{ Theorem \ref{theorem_opt_n} reveals that exploiting a larger blocklength is beneficial for improving the secrecy throughput under given channel gains. 
This result is nontrivial in light of \cite{Farhat2018Secure} where there exists a critical value of the blocklength, instead of the maximal one, that can achieve the maximal secrecy throughput. The main reason behind the two different results lies in that, Bob's instantaneous CSI is available here and is adequately exploited, and the codeword rate will not exceed Bob's channel capacity under the on-off policy such that using a larger blocklength can always lower the decoding error probability for Bob. Combined with Lemma \ref{lemma_opt_re_n}, it can be seen that increasing the blocklength improves reliability and secrecy simultaneously, thus making the secrecy throughput higher. However, this can no longer be promised in \cite{Farhat2018Secure} where the instantaneous CSI of the main channel is unknown, and using a larger blocklength might degrade the reliability once the codeword rate exceeds Bob's channel capacity, just as implied in Lemma \ref{lemma_opt_re_n}.}
Revisiting \eqref{max_st_ad},
since $\mu^*$ in the lower limit of the integral decreases with $n$ (see \eqref{opt_mu_ad} where $R_e^*$ decreases with $n$), it is clear that the global optimal blocklength that maximizes $\mathcal{T}_{\rm A}$ is also $n^* = N$.
\subsubsection{Solving $R_s$}
Substituting the derived optimal $R_e^*$, $\mu^*$, and $n^*$ into \eqref{pe} yields the maximal $p_s$,
and then the optimal $R_s^*$ can be determined by solving the following problem:
\begin{subequations}\label{st_max_rs}
	\begin{align}
	\max_{R_s}& ~\mathcal{T}_{\rm A}(\eta)= R_s\left[1 - Q\left(\frac{C_b-R_s-R_e^*}{\sqrt{V_b/N}}\right)\right]\label{st_max_rs_c1}	\\
	~~{\rm s.t.}&~~0<R_s\le C_b - R_e^*.
	\end{align}
\end{subequations}
\begin{theorem}\label{theorem_opt_rs}
	$\mathcal{T}_{\rm A}(\eta)$ in \eqref{st_max_rs} is a concave function of $R_s$, and its maximal value is achieved at 
\begin{align}\label{opt_rs}
R_s^*
=\begin{cases}
C_b - R_e^*, &  \eta\le\frac{\gamma_b^{\circ}}{P_b},\\
R_s^{\circ}, & \rm otherwise,
\end{cases}
\end{align}
where $\gamma_b^{\circ}\in\left(\sqrt{\frac{1}{2}+\sqrt{\frac{1}{4}+\frac{\pi}{2N}}}-1,e^{\sqrt{\frac{\pi}{2N}}+R_e^*\ln 2}-1\right)$ is the unique root $\gamma_b>0$ that satisfies $ C_b-\sqrt{\frac{\pi V_b}{2N}}=R_e^*$, and $R_s^{\circ}$ is the unique zero-crossing $R_s<C_b - R_e^*$ of the derivative 
\begin{equation}\label{dTs}
\frac{d\mathcal{T}_{\rm A}(\eta)}{dR_s} = 1 - Q\left(\frac{C_b-R_s-R_e^*}{\sqrt{V_b/N}}\right)-\frac{R_s\sqrt{N}}{\sqrt{2\pi V_b}}e^{-\frac{\left(C_b-R_s-R_e^*\right)^2}{2V_b/N}}.
\end{equation}
\end{theorem}
\begin{proof}
	Please refer to Appendix \ref{appendix_theorem_opt_rs}.
\end{proof}

Theorem \ref{theorem_opt_rs} presents an optimal secrecy rate $R_s^*$ that differs from the one for  infinite blocklength with $N\rightarrow\infty$, where in the latter employing the maximal achievable secrecy rate $R_s^* = C_b -R_e^*$ is always optimal for secrecy throughput improvement. 
The fundamental reason behind such difference lies in the decoding failure caused by finite blocklength. Specifically,  when the quality of the main channel is poor (i.e., a small $\eta$) or when a large rate redundancy $R_e^*$ is required, e.g., due to a high average SNR of Eve or a {stringent} secrecy requirement, the successful decoding probability $p_s$ is initially small and decreases slowly with $R_s$. 
In this case, the secrecy throughput {improvement} is mainly bottlenecked by $R_s$, and hence it is necessary to choose the maximal secrecy rate $R_s^* = C_b -R_e^*$. Otherwise, $p_s$ is initially large but drops rapidly with $R_s$, thus dramatically degrading the secrecy throughput. Therefore, a relatively small $R_s$ is supposed to be chosen to strike a good balance between the decoding and throughput performance.

The optimal secrecy rate $R_s^{*}\le C_b - R_e^*$ in \eqref{opt_rs} can be obtained efficiently using the Newton's method, despite its implicit form. The following corollaries further give a closed-form {asymptotically tight} lower bound $R_s^{L}$ on $R_s^{*}$ and provide useful insights into the behavior of $R_s^{*}$.
\begin{corollary}\label{corollary_ad1}
	The optimal secrecy rate $R_s^*$ in \eqref{opt_rs} satisfies
	\begin{equation}\label{lower_bound}
	R_s^{*}\ge R_s^{L}\triangleq C_b - R_e^* - \sqrt{\frac{2V_b}{N}\ln\left(\frac{1}{2}+\frac{C_b-R_e^*}{\sqrt{2\pi V_b/N}}\right)}.
	\end{equation}
\end{corollary}
\begin{proof}
	The result follows by finding a lower bound on $\frac{d\mathcal{T}_{\rm A}(\eta)}{dR_s}$ in \eqref{dTs} applying the inequalities $Q(x)\le\frac{1}{2}e^{-x^2/2}$ and $R_s\le C_b-R_e^*$ and then setting the resultant lower bound to zero.
\end{proof}

The term $ \sqrt{\frac{2V_b}{N}\ln\left(\frac{1}{2}+\frac{C_b-R_e^*}{\sqrt{2\pi V_b/N}}\right)}$ in \eqref{lower_bound} is interpreted as the secrecy rate loss arisen from finite blocklength. This term vanishes as $N\rightarrow\infty$ or $R_e^*\rightarrow C_b  -\sqrt{\frac{\pi V_b}{2N}}$, and accordingly $R_s^{*}$ approaches $C_e-R_e^*$. In this sense, the lower bound $R_s^{L}$ can be employed as a computational convenient alternative to the optimal $R_s^*$, particularly for the large blocklength scenarios. 

\begin{corollary}\label{corollary_ad2}
	The optimal secrecy rate $R_s^*$ monotonically increases with the channel gain $\eta$.
\end{corollary}
\begin{proof}
	It is proved that $\frac{C_b-R_s-R_e^*}{\sqrt{V_b}}$ in \eqref{dTs} increases with $\eta$ such that $\frac{d\mathcal{T}_{\rm A}(\eta)}{dR_s}$ increases with $\eta$. Then, using the derivative rule for implicit functions with $\frac{d\mathcal{T}_{\rm A}(\eta)}{dR_s^*}=0$ reaches $\frac{dR_s^*}{d\eta}>0$.
\end{proof}

{ Fig. \ref{ADA_ST_RS} depicts  secrecy throughput $\mathcal{T}_{\rm A}(\eta)$ versus secrecy rate $R_s$ for different blocklength $N$ and channel gain $\eta$.
The concavity of $\mathcal{T}_{\rm A}(\eta)$ on $R_s$ given by Theorem \ref{theorem_opt_rs} is well verified. Specifically, $\mathcal{T}_{\rm A}(\eta)$ first increases and then decreases with $R_s$, and there exists an optimal $R_s^*$ that maximizes $\mathcal{T}_{\rm A}(\eta)$.
It is also found that $\mathcal{T}_{\rm A}(\eta)$ almost linearly increases with $R_s$ at first, since the throughput loss due to decoding error is negligible.
	Note that the curves in the figure are cut in different points which represent different values of the maximal achievable secrecy rate $R_s^{\max}$ for different $N$ and $\eta$, and it is obvious that $R_s^{\max}$ increases with $N$ and $\eta$. As $\eta$ grows, 
	$\mathcal{T}_{\rm A}(\eta)$ improves significantly and the corresponding optimal $R_s^*$ increases, which validates Corollary \ref{corollary_ad2}. The underlying reason is that, when the main channel quality improves, choosing a larger $R_s$ contributes more to improving $\mathcal{T}_{\rm A}(\eta)$ compared with increasing the successful decoding probability $p_s$ (by lowering $R_s$).
	In addition, as proved in Theorem \ref{theorem_opt_n},  $\mathcal{T}_{\rm A}(\eta)$ increases with $N$. It is also proved that the optimal $R_s^*$ increases with $N$ as $\eta\rightarrow\infty$. However, it is no longer true when $\eta$ is too small, e.g., $\eta=3$ dB. This is because, for a low channel quality, the decoding performance becomes a key restricting factor on throughput improvement, and hence $R_s$ should be decreased to ensure a large $p_s$ as $N$ increases. 
	Moreover, the secrecy throughput obtained with the lower bound $R_s^{L}$ in Corollary \ref{corollary_ad1} approaches closely the optimal one particularly when $N$ is sufficiently large, which demonstrates the usefulness of the lower bound. 
	}

\begin{figure}[!t]
	\centering
	\includegraphics[width = 3.7in]{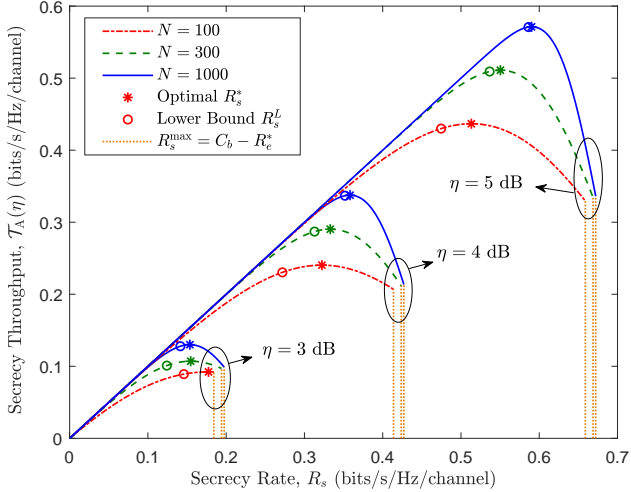}
	\caption{$\mathcal{T}_{\rm A}(\eta)$ vs. $R_s$ for different $N$ and $\eta$, with $P_b = 0$ dB, $\Gamma_e = 0$ dB, and $\delta = 0.2$.} 
	\label{ADA_ST_RS}
\end{figure}

\subsection{Non-Adaptive Optimization Scheme}

This section devises a non-adaptive optimization scheme where the parameters $\mu$, $R_s$, $R_e$, and $n$ are designed based on the statistical CSI of the main channel and remain unchanged during the transmission period. Such a non-adaptive scheme can be {computed} off-line, which significantly lowers the complexity compared with an adaptive one.

Since all the parameters are independent of the channel gain $\eta$, the problem of maximizing the secrecy throughput in \eqref{def_st} can be recast as follows:
\begin{equation}\label{st_max_non}
\max_{\mu, R_e, R_s, n} ~\mathcal{T}_{\rm N}=R_s\bar p_s
~~{\rm s.t.} 
~~\eqref{st_max_c1}-\eqref{st_max_c3},
\end{equation}
where $\bar p_s=\int_{P_b\mu}^{\infty}p_sf_{\gamma_b}(\gamma)d\gamma$ denotes the average successful decoding probability.

The above problem can be {handled via} similar steps for its adaptive counterpart in Sec. \ref{ada_single}. To begin with, in order to increase $p_s$ for a given $R_s$, a minimal rate redundancy $R_e$ should be chosen while satisfying the secrecy constraint $\mathcal{O}_e\le\delta$. Hence, the optimal $R_e^*$ is given in \eqref{opt_re_ad}. It can be inferred from \eqref{st_max_non} that a smaller transmission threshold $\mu$ can produce a larger $\mathcal{T}_{\rm N}$. Nonetheless,  $C_b\ge\log_2(1+P_b\mu)\ge R_s+R_e^*$ must be ensured, since otherwise there would always exist a transmission initiated when $\eta>\mu$ while violating the reliability constraint \eqref{st_max_c1}. Consequently, the optimal $\mu^*$ for a fixed $R_s$ is given by
\begin{equation}\label{opt_mu_non}
\mu^* = \frac{2^{R_s+R_e^*}-1}{P_b}.
\end{equation}
Note that in order to support a constant secrecy rate $R_s$, the optimal on-off threshold $\mu^*$ for the non-adaptive scheme is generally larger than that of the adaptive one as given in \eqref{opt_mu_ad}.
On the other hand, the optimal $\mu^*$ monotonically decreases with $\delta$ and $n$, which is similar to the adaptive case. That is to say, the transmission condition can be relaxed when facing a looser secrecy requirement or using a larger blocklength.

Substituting $R_e^*$ and $\mu^*$ into $\bar p_s$ and invoking the {approximation} of $Q$-function in \eqref{app_q} yields
\begin{align}\label{st_non}
\bar p_s& =\int_{P_b\mu^*}^{\infty}\left[1- \Xi(\gamma_b,n,R_s+R_e^*)\right]f_{\gamma_b}(\gamma)d\gamma\nonumber\\
& \stackrel{\mathrm{(a)}}= 1-\mathcal{F}_{\gamma_b}(\theta_b^2)\int_{\theta_b^2}^{\tau^u_b}\left(\frac{1}{2}-\frac{\beta}{\theta_b}(\gamma - \theta_b^2)\right)f_{\gamma_b}(\gamma)d\gamma\nonumber\\
&\stackrel{\mathrm{(b)}}= 1-\frac{1}{2}\mathcal{F}_{\gamma_b}(\theta_b^2)-\frac{\beta}{\theta_b}\int_{\theta_b^2}^{\tau^u_b}\mathcal{F}_{\gamma_b}(\gamma)d\gamma,
\end{align}
where $\mathrm{(a)}$ is due to $\theta_b = \sqrt{P_b\mu^*}=\sqrt{2^{R_s+R_e^*}-1}$, $\beta= \frac{\sqrt{n}}{2\pi}$, and  $\tau^u_b=\theta_b^2+\frac{\theta_b}{2\beta}$, and $\mathrm{(b)}$ stems from the {use of} partial integration.
 With \eqref{st_non}, the problem of maximizing $\mathcal{T}_{\rm N}$ over $n$ and $R_s$ can be equivalently transformed as below:
 \begin{subequations}\label{st_max_re}
 \begin{align}
 \max_{\beta, \theta_b} &~~\mathcal{T}_{\rm N}=\left[\log_2(1+\theta_b^2)-R_e^*\right]{\bar p_s}\\
 ~~{\rm s.t.} &
 ~~\frac{1}{2\pi}\le \beta \le \frac{\sqrt{N}}{2\pi}, ~ \theta_b>\sqrt{2^{R_e^*}-1}.
 \end{align}
 \end{subequations}
\begin{theorem}\label{theorem_opt_n_non}
 $\mathcal{T}_{\rm N}$ in \eqref{st_max_re} is a monotonically increasing function of $\beta$ or $n$.
\end{theorem}
\begin{proof}
The result follows by proving that 
\begin{align}
\frac{d\mathcal{T}_{\rm N}}{d\beta}&=
-\frac{dR_e^*}{d\beta}\bar p_s+\left[\log_2(1+\theta_b^2)-R_e^*\right]\frac{d\bar p_s}{d\beta}\nonumber\\
& \stackrel{\mathrm{(a)}}>\left[\log_2(1+\theta_b^2)-R_e^*\right]\frac{d\bar p_s}{d\beta} \stackrel{\mathrm{(b)}}>0,
\end{align}
where $\mathrm{(a)}$ is due to $\frac{dR_e^*}{dn}<0$ from \eqref{dRe}, and $\mathrm{(b)}$ follows from $\frac{d\bar p_s}{d\beta}=\frac{1}{\theta_b}\int_{\theta_b^2}^{\tau^u_b}\left[\mathcal{F}
_{\gamma_b}(\tau^u_b)-\mathcal{F}_{\gamma_b}(\gamma)\right]d\gamma>0$ as $\mathcal{F}_{\gamma_b}(\gamma)$ is an increasing function of $\gamma$.
\end{proof}

Theorem \ref{theorem_opt_n_non} suggests that Alice should use the maximal blocklength to maximize the secrecy throughput for the non-adaptive scheme, regardless of other parameters, i.e., the globally optimal blocklength is $n^* = N$. More importantly, this conclusion {holds} for any distribution of $\gamma_b$.

Substituting the CDF $\mathcal{F}_{\gamma_b}(\gamma) = 1 - e^{-{\gamma}/{\Gamma_b}}$ into \eqref{st_non} yields
\begin{equation}\label{st_theta}
\mathcal{T}_{\rm N} =\frac{1}{2} \left[\log_2(1+\theta_b^2)-R_e^*\right]\left[1+Y(\theta_b)\right]e^{-\frac{\theta_b^2}{\Gamma_b}},
\end{equation} 
where $Y(\theta_b) = \frac{2\beta\Gamma_b}{\theta_b}(1-e^{-\frac{\theta_b}{2\beta\Gamma_b}})>0$. 
The optimal $\theta_b^*$ that maximizes $\mathcal{T}_{\rm N}$ is provided below.
\begin{theorem}\label{theorem_opt_rs_non}
 $\mathcal{T}_{\rm N}$ in \eqref{st_theta} is  first-increasing-then-decreasing w.r.t. $\theta_b$; the optimal $\theta_b^*$ maximizing $\mathcal{T}_{\rm N}$ is the unique root $\theta_b>\sqrt{2^{R_e^*}-1}$ of $G(\theta_b)=0$, where $G(\theta_b)$ is a decreasing function of $\theta_b$:
\begin{equation}\label{G}
G(\theta_b)=\frac{1+Y(\theta_b)}{\ln2}-\left[\log_2(1+\theta_b^2)-R_e^*\right]\frac{1+\theta_b^2}{\theta_b}g(\theta_b),
\end{equation}
with $g(\theta_b) = \left(\frac{1}{2\theta_b}+\frac{1}{4\beta\Gamma_b}+\frac{\theta_b}{\Gamma_b}\right)Y(\theta_b)+\frac{\theta_b}{\Gamma_b}-\frac{1}{2\theta_b}$.
\end{theorem}
\begin{proof}
	Please refer to Appendix \ref{appendix_theorem_opt_rs_non}.
\end{proof}

Based on Theorem \ref{theorem_opt_rs_non}, the optimal $\theta_b^*$ or secrecy rate $R_s^*=\log_2(1+(\theta_b^*)^2)-R_e^*$ can be efficiently calculated using a bisection search with $G(\theta_b)=0$, and thus the maximal $\mathcal{T}_{\rm N}^*$ can be obtained from \eqref{st_max_re}. The following corollaries demonstrate the behavior of $R_s^*$ w.r.t. to the average channel power gain  $\sigma_b^2=\frac{\Gamma_b}{P_b}$ and provide a closed-form approximation of $R_s^*$ at the large $\sigma_b^2$ regime.   
\begin{corollary}\label{corollary_non1}
	The optimal $R_s^*$ monotonically increases with $\sigma_b^2$.
\end{corollary}
\begin{proof}
	Following similar steps as the proof of Theorem \ref{theorem_opt_rs_non}, it can be verified that $G(\theta_b)$ in \eqref{G} increases with  $\sigma_b^2$ such that $\frac{dR_s^*}{d\sigma_b^2}=- \frac{\partial G(\theta_b)/\partial \sigma_b^2}{\partial G(\theta_b)/\partial R_s^*}>0$, which completes the proof.	
\end{proof}

 Corollary \ref{corollary_non1} suggests that a larger secrecy rate should be employed to boost {the} secrecy throughput when the quality of the main channel improves, {despite the fact that} it might deteriorate the decoding correctness {at Bob}. 

\begin{corollary}\label{corollary_non2}
	At the regime of $\sigma_b^2\rightarrow\infty$, the optimal secrecy rate $R_s^*$ is approximated by 
	\begin{align}
R_s^*
\approx R_s^{A}&= {\log_2(e)}{\mathcal{W}_0\left(\sigma_b^2 2^{-R_e^*}\right)}\nonumber\\
&\approx \log_2(\sigma_b^2) - R_e^* - \log_2\left[\ln (\sigma_b^2) - R_e^*\ln 2\right],
	\end{align}
	where $\mathcal{W}_0(x)$ is the Lambert's $W$ function \cite[Sec. 4.13]{Olver2010NIST} that satisfies $x = \mathcal{W}_0(x)e^{\mathcal{W}_0(x)}$.
\end{corollary}
\begin{proof}
	It is clear that $Y(\theta_b)\rightarrow 1$ and $g(\theta_b)\rightarrow\frac{2\theta_b}{\Gamma_b}$ as $\sigma_b^2\rightarrow\infty$. Substituting the results into \eqref{G} with $\theta_b^2 = 2^{R_s+R_e^*}-1$ and letting $G(\theta_b)=0$ produce the first approximation. The second approximation comes from the expansion of $\mathcal{W}_0(x)$ as $x\rightarrow\infty$ that $\mathcal{W}_0(x)\approx \ln x - \ln(\ln x)$. 
\end{proof}

{ Fig. \ref{NON_ST_RS} plots the secrecy throughput $\mathcal{T}_{\rm N}$ versus the secrecy rate $R_s$ for different values of the blocklength $N$ and the average channel gain $\sigma_b^2$. It can be seen that $\mathcal{T}_{\rm N}$ first increases and then decreases with $R_s$, which validates Theorem \ref{theorem_opt_rs_non}.
The optimal $R_s^*$ maximizing $\mathcal{T}_{\rm N}$ increases with $\sigma_b^2$, which verifies Corollary \ref{corollary_non1} well, and the reason behind is similar to that for Corollary \ref{corollary_ad2}. It can  also be observed that the optimal $R_s^*$ is almost impervious to different $N$. This is because, the optimal secrecy rate for the non-adaptive scheme only depends on the average successful decoding probability, and the averaging process softens the impact of the blocklength.
Theorem \ref{theorem_opt_n_non} is also confirmed, where it is found that $\mathcal{T}_{\rm N}$ increases with $N$.
In addition, the secrecy throughput with the approximate $R_s^{A}$ obtained in Corollary \ref{corollary_non2} is almost coincided with that of the optimal $R_s^*$, which demonstrates the {practicability} of the low-complexity approximation.}
\begin{figure}[!t]
	\centering
	\includegraphics[width = 3.6in]{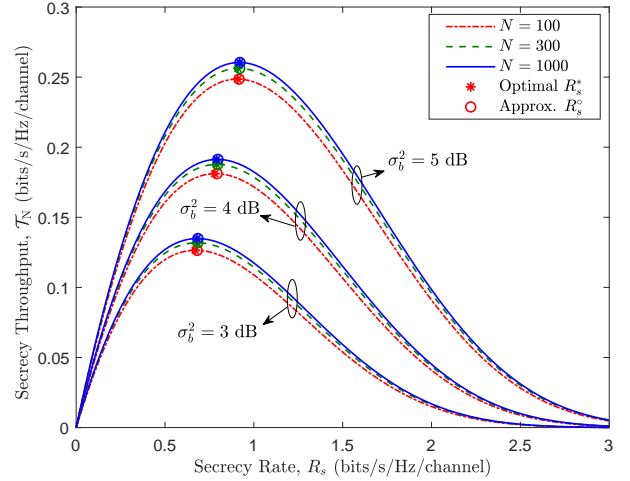}
	\caption{$\mathcal{T}_{\rm N}$ vs. $R_s$ for different $N$ and $\sigma_b^2$, with $P_b = 0$ dB, $\Gamma_e = 0$ dB, and $\delta = 0.2$. }
	\label{NON_ST_RS}
\end{figure}

Fig. \ref{MAX_ST_SINGLE} compares the secrecy throughput for adaptive and non-adaptive schemes with different blocklength $N$.
The left-hand-side figure depicts the maximal secrecy throughput $\mathcal{T}^*$, where $\mathcal{T}_{\rm A}^*$ for the adaptive case improves as $N$ increases whereas $\mathcal{T}_{\rm N}^*$ for the non-adaptive case almost remains unchanged. When the average channel gain $\sigma_b^2$ increases or the secrecy constraint becomes relaxed (i.e., a larger $\delta$), the maximal $\mathcal{T}^*$ for both schemes improves significantly, and the gap $\mathcal{T}_{\rm A}^*-\mathcal{T}_{\rm N}^*$ increases.
The right{-hand-side} figure illustrates the relative throughput gain $\Delta\mathcal{T}\triangleq \frac{\mathcal{T}_{\rm A}^*-\mathcal{T}_{\rm N}^*}{\mathcal{T}_{\rm N}^*}$ which reflects the superiority of the adaptive scheme over its non-adaptive counterpart. It is shown that $\Delta\mathcal{T}$ grows dramatically with $N$ but decreases with $\sigma_b^2$ and $\delta$. This suggests that the adaptive scheme is more preferred for some {\it unfavorable} scenarios, e.g., with a large blocklengh (large delay), a poor channel quality, or a {stringent} secrecy requirement; otherwise, the non-adaptive scheme could be an alternative choice owing to its low implementation complexity.
 
\begin{figure}[!t]
	\centering
	\includegraphics[width = 3.6in]{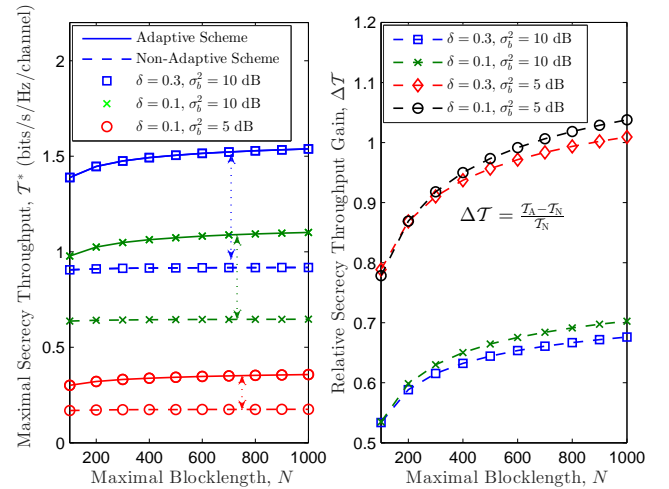}
	\caption{$\mathcal{T}^*$ and $\Delta\mathcal{T}$ vs. $N$ for different $\sigma_b^2$ and $\delta$, with $P_b = 0$ dB and $\Gamma_e = 0$ dB.}
	\label{MAX_ST_SINGLE}
\end{figure}

\section{Multi-Antenna Transmitter Scenario}
{When Alice is equipped with multiple antennas, she can intentionally transmit AN together with the information-bearing signal to degrade Eve's channel quality. Generally, the null-space AN scheme, in which the AN is injected uniformly in directions orthogonal to the main channel, is heuristically employed in the context of infinite blocklength \cite{Goel2008Guaranteeing}. 
The near-optimality of AN in terms of improving secrecy capacity for the multi-input single-output wiretap channel was first proved in \cite{Khisti2010Secure_I} from a rigorous information-theoretic perspective, and its degraded performance was later observed for the multi-input multi-output wiretap channel \cite{Khisti2010Secure_II}.
On the other hand, it was argued in \cite{Lin2013On} that distributing a certain proportion of AN in the direction of main channel can surprisingly gain a larger ergodic secrecy rate.
{When it comes to finite blocklength, since decoding failure might occur even when the codeword rate lies below the channel capacity, which is quite different from the infinite blocklength case, it is still unclear whether the null-space AN is optimal and how the optimal power allocation of the AN scheme should be determined for maximizing the secrecy throughput. To this end, this section focuses on the optimization of secrecy throughput with finite blocklength for the multi-antenna scenario, where the optimality of the null-space AN scheme will be identified first. }}

Considering a general scenario where the AN is not restricted to be orthogonal to the main channel, Alice's transmitted signal can be constructed in the form of 
\begin{equation}\label{an_signal}
\bm x = \sqrt{\phi P}{\bm w}\left(\sqrt{\alpha}s+\sqrt{1-\alpha}v\right)+\sqrt{\frac{(1-\phi)P}{M-1}}{\bm W}_{\bot}{\bm z},
\end{equation}
where ${\bm w}=\frac{\bm h_b^{\dag}}{\|\bm h_b\|}$ denotes the beamforming vector for the main channel, ${\bm W}_{\bot}$ denotes the $M\times(M-1)$ projection matrix onto the null space of $\bm h_b$ such that ${\bm h_b^{\mathrm{T}}}{\bm W}_{\bot}=\bm 0$, and the columns of $[{\bm w} ~{\bm W}_{\bot}]$ constitute an orthogonal basis; $s$, $v$, and ${\bm z}$ denote the information signal, the AN in the direction of ${\bm w}$, and the AN in the null space ${\bm W}_{\bot}$, with each element obeying $\mathcal{CN}(0,1)$; $\phi\in[0,1]$ represents the fraction of the total transmit power $P$ allocated to the direction of ${\bm w}$, and $\alpha\in[0,1]$ represents the power allocation ratio of the information signal to $\phi P$.
With \eqref{an_signal}, the received signal-to-interference-plus-noise ratios (SINRs) at Bob and Eve are  respectively 
\begin{align}
\label{sinr_bob}
\gamma_b &= \frac{\alpha\phi P_b\eta}{(1-\alpha)\phi P_b\eta+1},\\
\label{sinr_eve}
\gamma_e &= \frac{\alpha\phi P_e|\bm h_e^{\mathrm T}\bm w|^2}{(1-\alpha)\phi P_e|\bm h_e^{\mathrm T}\bm w|^2+\frac{(1-\phi)P_e\|\bm h_e^{\mathrm T}\bm W_{\bot}\|^2}{M-1}+1},
\end{align}
where $\eta = \|\bm h_b\|^2$.
The successful decoding probability $p_s$ and the information leakage probability $\mathcal{O}_e$ for the multi-antenna case are still given by \eqref{error_pro} and \eqref{oe}, respectively. The {corresponding} secrecy throughput optimization problem can be formulated as below:
	\begin{align}\label{st_max_multi_problem}
	\max_{\mu, R_e, R_s, n,\alpha,\phi} ~\mathcal{T}=\mathbb{E}_{\eta}\left[R_sp_s\right]
	~~{\rm s.t.} 
	~~\eqref{st_max_c1}-\eqref{st_max_c3}, ~0\le \alpha,\phi\le 1.
	\end{align}

{The following subsections will first detail the optimization procedure for both adaptive and non-adaptive schemes, and then briefly discuss the scenario of a multi-antenna Eve.}

\subsection{Adaptive Optimization Scheme}
This subsection optimizes the secrecy throughput $\mathcal{T}_{\rm A}$ by designing the parameters involved in problem \eqref{st_max_multi_problem} adaptively according to the instantaneous channel realization $\bm h_b$.

\subsubsection{Solving $R_e$}
Similar to the single-antenna case, the optimal rate redundancy is given by $R_e^* = \mathcal{O}_e^{-1}(\delta)$ with $\mathcal{O}_e$ in \eqref{oe}. Note that $R_e^*$ herein is a function of $\phi$ and $\alpha$. 

\subsubsection{Solving $\alpha$}
Resort to a function $\kappa(x,\alpha)\triangleq\frac{x\alpha}{x(1-\alpha)+1}$ defined in \cite{Wang2015Secrecy}, which increases with $x$ for $\alpha>0$. Then, the SINRs $\gamma_b$ in \eqref{sinr_bob} and $\gamma_e$ in \eqref{sinr_eve} can be reformulated as $\gamma_b(\phi,\alpha)=\kappa\left(\gamma_b(\phi,1),\alpha \right)$ and $\gamma_e(\phi,\alpha)=\kappa\left(\gamma_e(\phi,1),\alpha \right)$. 
Define $\Phi_e(\phi,\alpha)\triangleq 2^{R_e^*}-1$ as the SINR threshold for $\gamma_e(\phi,\alpha)$ such that $R_e^* = \log_2(1+\Phi_e(\phi,\alpha))$. Recalling the secrecy constraint $\mathcal{O}_e(\Phi_e;\theta,\alpha)=\delta$,  $\Phi_e(\phi, \alpha)$ is the $\delta$-upper quantile of $\gamma_e(\phi,\alpha)$ such that it also follows the form $\Phi_e(\phi,\alpha)=\kappa(\Phi_e(\phi,1),\alpha)$ \cite{Wang2015Secrecy}.
Hence, the condition {for guaranteeing} a positive secrecy rate is described as
\begin{align}\label{neq_positive}
\gamma_b(\phi,\alpha)>\Phi_e(\phi,\alpha)&\Rightarrow\kappa(\gamma_b(\phi,1),\alpha)>\kappa(\Phi_e(\phi,1),\alpha)\nonumber\\
&\Rightarrow\gamma_b(\phi,1)>\Phi_e(\phi,1)\nonumber\\
&\stackrel{\mathrm{(a)}}\Rightarrow \rho_b> \rho_e(\phi),
\end{align}
where $\rho_b\triangleq P_b\eta$, $\rho_e(\phi)\triangleq \frac{\Phi_e(\phi,1)}{\phi }$, and $\mathrm{(a)}$ is due to $\gamma_b(\phi,1)=\phi P\eta$.
Then, the threshold $\mu$ can be simply set as $\mu(\phi) = \frac{\rho_e(\phi)}{P_b}$ for any fixed $\phi$.
Revisiting \eqref{max_st_ad},  since $\mu(\phi)$ is independent of $\alpha$, the optimal $\alpha^*$ that maximizes $\mathcal{T}_{\rm A}$ can be obtained by maximizing $\mathcal{T}_{\rm A}(\eta)= R_sp_s$, where $p_s$ is defined in \eqref{pe} and can be rewritten as
\begin{equation}\label{ps_multi}
p_s = 1 - Q\left(\sqrt{n}\lambda_b\frac{\ln\lambda_b - \ln\lambda_e-R_s\ln2 }{\sqrt{\lambda_b^2-1}}\right),
\end{equation}
with $\lambda_b\triangleq 1+\kappa(\gamma_b(\phi,1),\alpha)>\lambda_e\triangleq 1+\kappa(\Phi_e(\phi,1),\alpha)>1$.
Although it is difficult to see how $p_s$ varies with $\alpha$ for a fixed $R_s<\log_2\frac{\lambda_b}{\lambda_e}$ as both $\lambda_b$ and $\lambda_e$ increase with $\alpha$,  the following theorem provides the optimal $\alpha^*$ that maximizes $\mathcal{T}_{\rm A}$.
\begin{theorem}\label{theorem_opt_alpha}
	$\alpha^*=1$ is optimal for maximizing the secrecy throughput $\mathcal{T}_{\rm A}$. 
\end{theorem}
\begin{proof}
	Please refer to Appendix \ref{appendix_theorem_opt_alpha}.
\end{proof}

Theorem \ref{theorem_opt_alpha} suggests that there is no need to inject the AN in the main channel direction for secrecy throughput improvement with finite blocklength. The reason is that, once the main channel quality suffices to guarantee $\lambda_b>\lambda_e$, a larger $\alpha$ can improve the term $\frac{\ln\lambda_b - \ln\lambda_e }{\sqrt{1-\lambda_b^{-2}}}$ in \eqref{ps_multi} which reflects the channel superiority of the main channel over the wiretap channel.

Define $\xi \triangleq \frac{\phi^{-1}-1}{M-1}$.  Substituting $\alpha^*=1$ into \eqref{sinr_bob} and \eqref{sinr_eve} yields the CDFs of $\gamma_b$ and $\gamma_e$:
\begin{align}
\label{CDF_rb}
\mathcal{F}_{\gamma_b}(\gamma) &= 1-e^{-\frac{\gamma}{\phi \Gamma_b}}\sum_{k=0}^{M-1}\frac{1}{k!}\left(\frac{\gamma}{\phi \Gamma_b}\right)^{k},\\
\label{CDF_re}
\mathcal{F}_{\gamma_e}(\gamma) &= 1-e^{-\frac{\gamma}{\phi \Gamma_e}}\left(1+\xi\gamma\right)^{1-M},
\end{align}

\subsubsection{Solving $\mu$}
The threshold $\mu(\phi) = \frac{\rho_e(\phi)}{P_b}$ mentioned in the last step is related to $\phi$. This step further determines the optimal $\mu^*$ which is independent of $\phi$ and $\eta$.
For tractability, consider an asymptotically large blocklength and {exploit the tail property} of the $Q$-function, then the information leakage probability $\mathcal{O}_e$ is approximated as \cite{Makki2016Wireless}
\begin{equation}\label{oe_app}
\mathcal{O}_e(\Phi_e)\approx  e^{-\frac{\Phi_e}{\phi \Gamma_e}}\left(1+\xi\Phi_e\right)^{1-M}.
\end{equation}
Fig. \ref{OE} shows that the approximate $\mathcal{ O}_e(\Phi_e)$ is extremely close to the exact value for quite a wide range of $\phi$, $M$, $n$, and $\Gamma_e$, and it then can be adopted to facilitate the subsequent analysis and optimization.
Revisiting $\rho_e(\phi)= \frac{\Phi_e(\phi)}{\phi }$ with $\mathcal{ O}_e(\Phi_e(\phi))=\delta$, the following lemma is obtained.
\begin{figure}[!t]
	\centering
	\includegraphics[width = 3.6in]{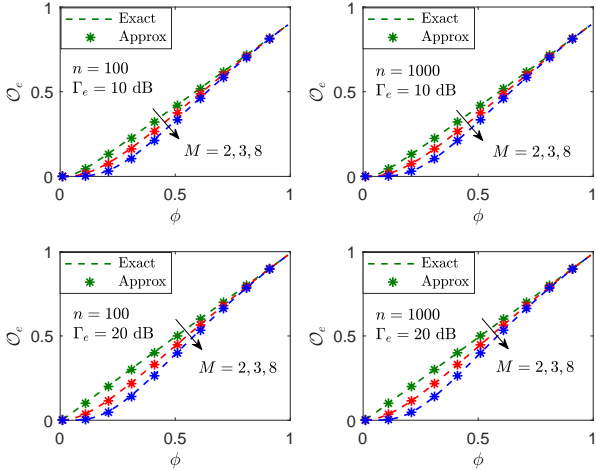}
	\caption{$\mathcal{O}_e$ vs. $\phi$ for different $M$, $n$, and $\Gamma_e$.}
	\label{OE}
\end{figure}

\begin{lemma}[{\cite{Wang2015Secrecy}}]\label{dchi}
$\rho_e(\phi)>0$, $\frac{d\rho_e(\phi)}{d\phi}=\frac{\rho_e(\phi)}{\left[1+\phi\rho_e(\phi)\xi\right]/\Gamma_e+1-\phi}>0$, and $\frac{d^2\rho_e(\phi)}{d\phi^2}>\frac{2}{\rho_e(\phi)}\left[\frac{d\rho_e(\phi)}{d\phi}\right]^2>0$.
\end{lemma}

Lemma \ref{dchi} indicates that $\rho_e(\phi)$ increases with $\phi$. It is observed from \eqref{neq_positive} that no positive $R_s$ can be achieved if $P_b\eta\le \rho_e(0)$. 
To avoid this, the optimal on-off threshold should be chosen as 
\begin{equation}\label{mu_multi}
\mu^* = \frac{\rho_e(0)}{P_b}.
\end{equation}

\subsubsection{Solving $n$}
This step gives the optimal blocklength $n^*$ that maximizes secrecy throughput. \begin{theorem}\label{theorem_opt_n_multi}
	$n^*=N$ is optimal for maximizing $\mathcal{T}_{\rm A}$ in \eqref{st_max_multi_problem}.
\end{theorem}
\begin{proof}
The proof is similar to that of Theorem \ref{theorem_opt_n}. According to \eqref{dqe2}, one only needs to prove that $f_{\gamma_e}(\tau^l_e)>f_{\gamma_e}(\tau^u_e)$. 
The PDF of $\gamma_e$ is calculated from \eqref{CDF_re} and is given by
\begin{equation}\label{PDF_re}
f_{\gamma_e}(\gamma) =\left(\frac{1}{\phi \Gamma_e\left(1+\xi\gamma\right)^{M-1}}+\frac{\xi(M-1)}{\left(1+\xi\gamma\right)^{M}}\right)e^{-\frac{\gamma}{\phi \Gamma_e}}.
\end{equation}
Apparently, $f_{\gamma_e}(\gamma)$  decreases with $\gamma$ such that $f_{\gamma_e}(\tau^l_e)>f_{\gamma_e}(\tau^u_e)$, which completes the proof.
\end{proof}

Theorem \ref{theorem_opt_n_multi} suggests that a multi-antenna transmitter also should adopt the maximal blocklength to maximize the secrecy throughput, regardless of the power allocation and code rates. This is validated by Fig. \ref{ADA_PHI_ST_RS_MULTI}, and the reason behind is similar to that of the single-antenna case.

\subsubsection{Solving $\phi$}
By now, the secrecy throughput $\mathcal{T}_{\rm A}(\eta)$ conditioned on $\eta$ is given by 
\begin{equation}\label{st_gamma_multi}
\mathcal{T}_{\rm A}(\eta)=R_s\left(1-Q\left[\sqrt{N}\lambda_b\frac{\ln\frac{\lambda_b}{\lambda_e}-R_s\ln2}{\sqrt{\lambda_b^{2}-1}}\right]\right),
\end{equation}
where $\lambda_b = 1+\phi\rho_b$ and $\lambda_e = 1+\phi\rho_e$ with  $\rho_b>\rho_e$. For notational simplicity, $\phi$ has been dropped from $\rho_e(\phi)$. Obviously, maximizing $\mathcal{T}_{\rm A}(\eta)$ is equivalent to maximizing the following function:
\begin{equation}\label{l-function}
L(\phi) =\frac{\lambda_b}{\sqrt{\lambda_b^{2}-1}} \left({\ln\frac{\lambda_b}{\lambda_e}-R_s\ln2}\right).
\end{equation}
\begin{theorem}\label{theorem_opt_phi}
$L(\phi)$ in \eqref{l-function} is a concave function of $\phi$, and the optimal $\phi^*$ maximizing $L(\phi)$ is
\begin{align}\label{opt_phi}
\phi^*
=\begin{cases}
1, &  \eta\ge\frac{\rho_b^{\circ}}{P_b}~{\rm and}~\frac{\rho_e(1)}{1+\rho_e(1)}<\frac{1}{1+\Gamma_e},\\
\phi^{\circ}, & \rm otherwise.
\end{cases}
f\end{align}
Here $\phi^{\circ}$ is the unique zero-crossing $\phi\in[0,1)$ of the following derivative: 
\begin{equation}\label{dL}
\frac{dL(\phi)}{d\phi}=\frac{(1-A_\phi){\lambda_b}-1}{\phi\sqrt{\lambda_b^{2}-1}}-\frac{({\lambda_b-1})\left({\ln{\lambda_b}-B_\phi}\right)}{\phi\left(\lambda_b^2-1\right)^{3/2}},
\end{equation}
where $A_\phi \triangleq \frac{\phi}{\lambda_e}\left(\rho_e+\phi\frac{d\rho_e}{d\phi}\right)$ and $B_\phi \triangleq \ln\lambda_e + R_s\ln2$ with $\frac{d\rho_e}{d\phi}$ given in Lemma \ref{dchi}, and $\rho_b^{\circ}$ is the unique root $\rho_b$ of the equation $X(\rho_b)=0$ with $X(\rho_b)$ given below: \begin{equation}\label{condition_phi_1}
X(\rho_b)=(1-A_1)(1+\rho_b)-1-\frac{\ln(1+\rho_b)-B_1}{2+\rho_b}.
\end{equation}
\end{theorem}
\begin{proof}
	Please refer to Appendix \ref{appendix_theorem_opt_phi}.
\end{proof}

Theorem \ref{theorem_opt_phi} shows that, the naive beamforming scheme without injecting any AN is optimal for maximizing the secrecy throughput only when the quality of the main channel is good enough and meanwhile the quality of the wiretap channel is poor or a high information leakage probability is acceptable. Using the derivative rule for implicit functions with \eqref{dL} proves that $\frac{d\phi^*}{dR_s}>0$, which suggests that in order to support a higher secrecy rate, a larger fraction of power should be allocated to the information signal although at the cost of a larger required rate redundancy.

For a robust design perspective, a worst-case scenario is considered by ignoring Eve's thermal noise, i.e., $\Gamma_e\rightarrow\infty$ in \eqref{oe_app}, such that $\rho_e = \frac{\Lambda}{1-\phi}$ with $\Lambda = (M-1)(\delta^{\frac{1}{1-M}}-1)$.
It is seen from \eqref{dL} that $\phi^*$ is a function of $\eta$ and $\delta$, and the monotonicity of $\phi^*$ is revealed as below. 
  \begin{corollary}\label{corollary_opt_phi}
  	For the worst case $\Gamma_e\rightarrow\infty$, the optimal power allocation $\phi^*$ is non-decreasing w.r.t. $\eta$ and $\delta$. Moreover, $\lim_{\eta\rightarrow\infty}\phi^* = \frac{1}{\sqrt{\Lambda}+1}$ 
and $\lim_{\delta\rightarrow 1}\phi^* = 1$.
\end{corollary}
\begin{proof}
	Please refer to Appendix \ref{appendix_corollary_opt_phi}.
\end{proof}

Corollary \ref{corollary_opt_phi} suggests that when the quality of the main channel improves (i.e., a larger $\eta$) or the secrecy requirement is relaxed (i.e., a larger $\delta$), it would be more appealing to use a higher signal power to promote the main channel than to increase the AN power to degrade the wiretap channel. This is because that the main channel becomes the {dominate} factor to the improvement of secrecy throughput. Different from Theorem \ref{theorem_opt_phi} where $\phi^*=1$ can be achieved,  the optimal $\phi^*$ here only can be increased up to $\frac{1}{\sqrt{\Lambda}+1}$ as $\eta\rightarrow\infty$ due to Eve's background noise being ignored. Besides, it is unsurprising that $\phi^*=1$ for $\delta=1$ since there is no secrecy requirement.

\subsubsection{Solving $R_s$}
For any given power allocation $\phi^*$, it can be proved that the secrecy throughput $\mathcal{T}_{\rm A}(\eta)$ is a concave function of the secrecy rate $R_s$ as done in Theorem \ref{theorem_opt_rs}. Hence, the optimal $R_s^*$ maximizing $\mathcal{T}_{\rm A}(\eta)$ is given by \eqref{opt_rs} and a closed-form lower bound on $R_s^{*}$ can be found in \eqref{lower_bound}. Eventually, problem \eqref{st_max_multi_problem} can be addressed via an alternating optimization (AO) method, which is summarized in Algorithm \ref{ao_algorithm}. 
In addition, at the high $\eta$ regime, the optimal $\phi^*$ is independent of $R_s$, and hence a global optimal pair $(\phi^*, R_s^*)$ is obtained for maximizing  $\mathcal{T}_{\rm A}(\eta)$. 
\begin{algorithm}[!t]
\caption{AO Algorithm for Solving Problem \eqref{st_max_multi_problem}}
\begin{algorithmic}[1]\label{ao_algorithm}
\STATE Initialize $k=1$, $\phi^{(0)}\in[0,1]$, $R_s^{(0)}\ge 0 $, and assign $\epsilon$ a sufficiently small positive value, e.g., $\epsilon = 10^{-10}$;
\STATE Input $\delta\in[0,1]$, $N\ge1$, and $P_b, \Gamma_e, \eta = \|\bm h_b\|^2>0$;
\STATE Calculate $\mu$ from \eqref{mu_multi} and  $\mathcal{T}_{\rm A}^{(0)}(\eta)= R_s^{(0)}p_s^{(0)}$;
\IF{$\eta<\mu$}
\STATE $\mathcal{T}_{\rm A}^{(k)}(\eta)\leftarrow0$;
\ELSE
\STATE\label{update} Update $\phi^{(k)}\leftarrow \phi^{(k-1)}$, $R_s^{(k)}\leftarrow R_s^{(k-1)}$;
\STATE Calculate $\rho_b^{\circ}$ from \eqref{condition_phi_1};
\IF{$P_b\eta\ge\rho_b^{\circ}$}
\STATE $\phi^{(k)}\leftarrow 1$;
\ELSE
\STATE Calculate $\phi^{(k)}$ from \eqref{dL};
\ENDIF
\STATE Calculate $R_s^{(k)}$ from \eqref{opt_rs};
\STATE\label{calculate_rs} Update $\mathcal{T}_{\rm A}^{(k)}(\eta)\leftarrow R_s^{(k)}p_s^{(k)}$;
\WHILE{$\left|\left[\mathcal{T}_{\rm A}^{(k)}(\eta)-\mathcal{T}_{\rm A}^{(k-1)}(\eta)\right]/{\mathcal{T}_{\rm A}^{(k-1)}(\eta)}\right|\ge\epsilon$}
\STATE Update $k \leftarrow k + 1$;
\STATE Repeat step \ref{update} to step \ref{calculate_rs};
\ENDWHILE
\ENDIF
	\STATE Output $\mathcal{T}_{\rm A}^{(k)}(\eta)$
\end{algorithmic}
\end{algorithm}

Fig. \ref{ADA_PHI_ST_RS_MULTI} illustrates the optimal power allocation $\phi^*$ and the corresponding maximal secrecy throughput $\mathcal{T}_{\rm A}(\eta)$ for varying secrecy rate $R_s$. The maximal $\mathcal{T}_{\rm A}(\eta)$ is concave on $R_s$, which guarantees the global optimality of the solution and the convergence of the proposed AO algorithm. The optimal $\phi^*$ increases with $\eta$ and $R_s$, which verifies Corollary \ref{corollary_opt_phi}. { In addition, the curves of $\phi^*$ and $\mathcal{T}_{\rm A}(\eta)$ are truncated after $R_s$ exceeds some critical values. This can be explained similarly as that of Fig. \ref{ADA_ST_RS}.}  
It is shown that $\phi^*$ increases with the blocklength $N$, although slightly. This is because, increasing $N$ will mildly decrease the information leakage probability $\mathcal{O}_e$, thus allowing a larger portion of power to be devoted to transmitting the information-bearing signal.  

\begin{figure}[!t]
	\centering
	\includegraphics[width = 3.6in]{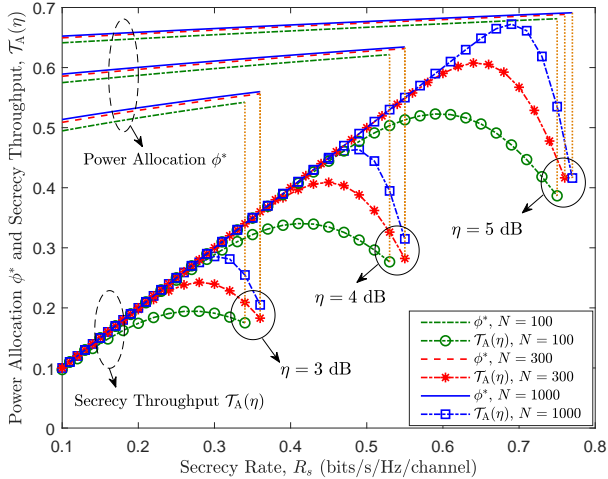}
	\caption{Optimal $\phi^*$ and  $\mathcal{T}_{\rm A}(\eta)$ vs. $R_s$ for different $N$ and $\eta$, with $M=4$, $P_b = 0$ dB, $\Gamma_e = 0$ dB, and $\delta = 0.2$.}
	\label{ADA_PHI_ST_RS_MULTI}
\end{figure}

\subsection{Non-Adaptive Optimization Scheme}
This subsection examines the secrecy throughput maximization through a non-adaptive design manner for the multi-antenna transmitter case. The problem can be formulated as
\begin{subequations}\label{st_max_non_multi}
	\begin{align}
\max_{\mu, R_e, R_s, n,\alpha,\phi}& ~\mathcal{T}_{\rm N}=R_s\bar p_s\\
~~{\rm s.t.} 
~&~\eqref{st_max_c1}-\eqref{st_max_c3}, ~0\le \alpha,\phi\le 1,
\end{align}
\end{subequations}
where $\bar p_s$ is the average successful decoding probability.

The basic idea to solve problem \eqref{st_max_non_multi} is similar to that of problem \eqref{st_max_non}. Again, the optimal rate redundancy is $R_e^* = \mathcal{O}_e^{-1}(\delta)$ with $\mathcal{O}_e$ given in \eqref{oe}.
For the adaptive case, it is known from \eqref{neq_positive} that $\kappa(\gamma_b(\phi,1),\alpha)>\kappa(\Phi_e(\phi,1),\alpha)\Rightarrow\eta>\mu(\phi)=\frac{\rho_e(\phi)}{P_b}$ suffices to guarantee a positive secrecy rate $R_s$ with the threshold $\mu(\phi)$ independent of $\alpha$, and then $\alpha^*=1$ is optimal for secrecy throughput maximization. As for the non-adaptive one, supporting a certain secrecy rate $R_s$ requires that $1+\kappa(\gamma_b(\phi,1),\alpha)>2^{R_s}(1+\kappa(\Phi_e(\phi,1),\alpha))$ which is further transformed to 
\begin{equation}\label{mu_multi_non}
\eta>\mu(\phi)=\frac{1}{\phi P_b}\frac{1}{\frac{\alpha}{2^{R_s}(1+\kappa(\Phi_e^*(\phi,1),\alpha))-1}-1+\alpha}.
\end{equation}
Although $\mu(\phi)$ herein depends on $\alpha$, it is proved that $\mu(\phi)$ monotonically decreases with $\alpha$. Hence, $\alpha^*=1$ is still throughput-optimal for the non-adaptive case. Accordingly, the optimal threshold is $\mu^*= \frac{2^{R_s}(1+\phi\rho_e)}{\phi P_b}$. {Similar to the proof of} Theorem \ref{theorem_opt_n_non}, using the maximal blocklength is optimal for maximizing secrecy throughput, regardless of the distribution of $\gamma_b$. Hence, the optimal blocklength is $n^*=N$. Afterwards, the secrecy throughput is calculated from \eqref{st_non}:
\begin{align}\label{st_multi}
\mathcal{T}_{\rm N} &= R_s\bar p_s = R_s\left[1-\frac{1}{2}\mathcal{F}_{\gamma_b}(\theta_b^2)-\frac{\beta}{\theta_b}\int_{\theta_b^2}^{\tau^u_b}\mathcal{F}_{\gamma_b}(\gamma)d\gamma\right]\nonumber\\&\stackrel{\mathrm{(a)}}= R_s\left[\frac{\bar\Gamma(M,\varrho_1)}{2}+\frac{\phi \Gamma_b\beta}{\theta_b}\Delta\Gamma\right],
\end{align}
{ where $\mathrm{(a)}$ holds by invoking the CDF $\mathcal{F}_{\gamma_b}(\gamma)$ of $\gamma_b$ in \eqref{CDF_rb} and computing the integral, with $\Delta\Gamma \triangleq \sum_{k=0}^{M-1}\left[\bar\Gamma(k+1,\varrho_1)-\bar\Gamma(k+1,\varrho_2\right] ) $ and $\bar \Gamma(m+1,x) \triangleq\sum_{k=0}^m\frac{x^ke^{-x}}{k!}$ being the regularized upper
incomplete gamma function, with $\varrho_1 = \frac{\theta_b^2}{\phi \Gamma_b}$, $\varrho_2 = \frac{\theta_b^2}{\phi \Gamma_b}+\frac{{\theta_b}}{2\beta\phi \Gamma_b}$, $\theta_b=\sqrt{2^{R_s+R_e^*}-1}$, and $\beta=\frac{\sqrt{N}}{2\pi}$.}
Differentiating $\mathcal{T}_{\rm N}$ w.r.t. $\phi$ yields 
\begin{align}
\frac{d\mathcal{T}_{\rm N}}{d \phi} =& R_s\bigg[\frac{\varpi_1\varrho_1^{M}e^{-\varrho_1}}{2(M-1)!} +\frac{\phi \Gamma_b\beta\varpi_2\Delta\Gamma}{\theta_b}+\beta\theta_b \varpi_1\Gamma(M,\varrho_1)\nonumber\\
&-\left(\beta\theta_b \varpi_1+\frac{\varpi_2}{2}\right) \bar\Gamma(M,\varrho_2)\bigg],
\end{align}
where $\varpi_1 = \frac{1}{\phi}-\frac{2^{R_s}}{\theta_b^2}\frac{d\lambda_e}{d\phi}$ and $\varpi_2 =\frac{1}{\phi}- \frac{2^{R_s}}{2\theta_b^2}\frac{d\lambda_e}{d\phi}$ with $\lambda_e$ given in \eqref{st_gamma_multi}. It is verified that the derivative $\frac{d\mathcal{T}_{\rm N}}{d \phi}$ is monotonically decreasing with $\phi$. In other words, for a fixed $R_s$, the optimal $\phi^*$ that maximizes $\mathcal{T}_{\rm N}$ is unique, which is $\phi^*=1$ if $\frac{d\mathcal{T}_{\rm N}}{d \phi}|_{\phi=1}>0$ or otherwise satisfies $\frac{d\mathcal{T}_{\rm N}}{d \phi}=0$. Likewise, it is confirmed that the derivative
\begin{align}
\frac{d\mathcal{T}_{\rm N}}{d R_s} =&
\frac{\bar\Gamma(M,\varrho_1)}{2} +\frac{\phi \Gamma_b\beta}{\theta_b}\Delta\Gamma-\frac{\lambda_eR_s2^{R_s}\ln2}{\theta_b}\bigg[{\beta}\bar\Gamma(M,\varrho_1)\nonumber\\
&+\frac{\phi \Gamma_b\beta \Delta\Gamma}{2\theta_b^2}+\frac{\varrho_1^Me^{-\varrho_1}}{2\theta_b(M-1)!}-\left({\beta}+\frac{1}{4\theta_b}\right)\bar\Gamma(M,\varrho_2)\bigg]
\end{align}
is first positive and then negative with increasing $R_s$, and the unique optimal $R_s^*$ maximizing $\mathcal{T}_{\rm N}$ can be calculated via a bisection method with the equation $\frac{d\mathcal{T}_{\rm N}}{d R_s}=0$. 

\begin{figure}[!t]
	\centering
	\includegraphics[width = 3.6in]{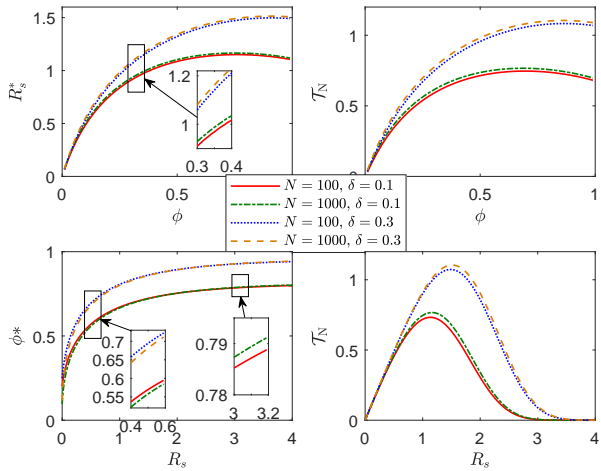}
	\caption{Above: $R_s^*$ and  $\mathcal{T}_{\rm N}$ vs. $\phi$; Bottom: $\phi^*$ and  $\mathcal{T}_{\rm N}$ vs. $R_s$; for different $N$ and $\delta$, with $M=4$, $\Gamma_b = 3$ dB and $\Gamma_e = 0$ dB.}
	\label{NON_ST_RS_PHI_MULTI}
\end{figure}
The monotonicity of $\mathcal{T}_{\rm N}$ w.r.t. $\phi$ and $R_s$ is verified in Fig. \ref{NON_ST_RS_PHI_MULTI}, where,
similar to Fig. \ref{ADA_PHI_ST_RS_MULTI}, $\mathcal{T}_{\rm N}$ is given with the optimal $\phi^*$ or $R_s^*$. This implies that the global maximal $\mathcal{T}_{\rm N}$ is practically achieved even by alternatively solving the optimal $\phi$ and $R_s$. 
As expected, $\mathcal{T}_{\rm N}$ improves with a larger blocklength $N$ and a looser secrecy constraint (a larger $\delta$).
It is found that $R_s^*$ first increases and then might decrease with $\phi$, which means that a moderate $R_s$ is desired to balance the decoding and throughput performance. On the other hand, a larger $\phi^*$ is required to support an increasing $R_s$.
It is also shown that $R_s^*$ for a fixed $\phi$ increases with $N$, since a larger $N$ improves the decoding performance which then affords a larger $R_s$. Nevertheless, $\phi^*$ decreases with $N$ in the low $R_s$ regime whereas increases with $N$ in the high $R_s$ regime. It can be explained as follows: for a low $R_s$, the rate redundancy $R_e$ has a great impact on the decoding performance, and hence the AN power should be increased as $N$ increases to better combat the eavesdropper; in contrast, for a large $R_s$, the decoding correctness is more affected by the main channel quality, which requires a larger signal power to maintain a high decoding probability.

%
\begin{proposition}\label{proposition_high}
	For the high average channel gain  $\Gamma_b\rightarrow\infty$, $\mathcal{T}_{\rm N}$ in \eqref{st_multi} is approximated as 
	\begin{equation}\label{st_high}
	\lim_{\Gamma_b\rightarrow\infty}\mathcal{T}_{\rm N} = R_s\left(1-\frac{\varrho_1^M}{2M!}\right).
	\end{equation}
\end{proposition}
\begin{proof}
	Please refer to Appendix \ref{appendix_proposition_high}.
\end{proof}

Proposition \ref{proposition_high} shows that for a high average channel gain, the secrecy throughput becomes independent of the blocklength.
In consequence, the optimal $\phi^*$ and $R_s^*$ maximizing $\mathcal{T}_{\rm N}$ in \eqref{st_high} admit the following closed-form approximations \cite[(19), (20)]{Zhang2013Design}
\begin{align}
\lim_{\Gamma_b\rightarrow\infty}\phi^*&= \frac{1}{\sqrt{\Lambda}+1},~~~\\
\lim_{\Gamma_b\rightarrow\infty}R_s^* &= \frac{1}{M\ln2}\left[\mathcal{W}_0\left(\frac{2\exp(1) M!\Gamma_b^M}{(\sqrt{\Lambda}+1)^{2M}}\right)-1\right].
\end{align}

Fig. \ref{MAX_ST_MM_MULTI} illustrates the influence of the number of transmit antennas $M$ on the maximal secrecy throughput $\mathcal{T}^*$ for both adaptive and non-adaptive schemes and the relative secrecy throughput gain $\Delta\mathcal{T}= \frac{\mathcal{T}_{\rm A}^*-\mathcal{T}_{\rm N}^*}{\mathcal{T}_{\rm N}^*}$.
It is not surprising that deploying more transmit antennas can significantly improve the secrecy throughput for both schemes. Similar to the observation in Fig. \ref{MAX_ST_SINGLE}, both $\mathcal{T}_{\rm A}^*$ and $\mathcal{T}_{\rm N}^*$ increase with $\delta$ and $N$, but the benefit to $\mathcal{T}_{\rm N}^*$ brought by a larger $N$ is nearly negligible. The right-hand-side subgraph shows that $\Delta\mathcal{T}$ drops sharply as $M$ increases but grows for a larger $N$ and a smaller $\delta$. This indicates that the superiority of the adaptive scheme over its non-adaptive counterpart is more pronounced for the scenarios requiring a large blocklengh, having few transmit antennas, suffering from a stringent secrecy constraint, etc; otherwise, the non-adaptive scheme might be appealing because of the low-complexity off-line design. 
\begin{figure}[!t]
	\centering
	\includegraphics[width = 3.6in]{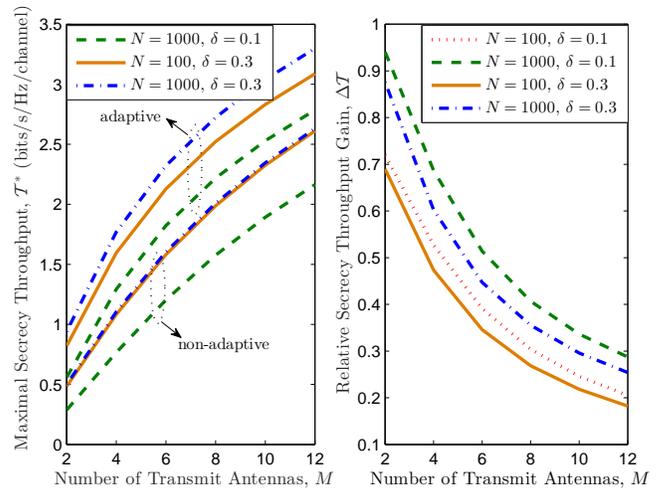}
	\caption{$\mathcal{T}^*$ and $\Delta\mathcal{T}$ vs. $M$ for different $N$ and $\delta$, with $\Gamma_b = 3$ dB and $\Gamma_e = 0$ dB.}
	\label{MAX_ST_MM_MULTI}
\end{figure}

{\subsection{A Note on Multi-Antenna Eve}
	This subsection examines the secure transmission in the presence of an Eve with $M_e$ antennas. Assume that Eve employs the minimum mean-squared error (MMSE) receiver, and then the CDF of Eve's SINR under the null-space AN scheme can be given as \cite{Zheng2017Safeguarding}:
	\begin{equation}\label{cdf_eve}
	\mathcal{F}_{\gamma_{e}}(x)=1-e^{-\frac{x}{\phi P_{e}} } \sum_{n=1}^{M_{e}} \frac{A_{n}(x)}{(n-1) !}\left(\frac{x}{\phi P_{e}}\right)^{n-1},
	\end{equation}
	where 
	\begin{align}
	A_{n}(x)
	=\begin{cases}
	1, &M_{e} \geq M-1+n,\\
	\frac{\sum_{m=0}^{M_{e}-n}\binom{M-1}{m}(\xi x)^{m}}{(1+\xi x)^{M-1}}, & M_{e}<M-1+n.\\
	\end{cases}
	\end{align}
	The information leakage probability $\mathcal{O}_e$ is obtained by substituting \eqref{cdf_eve} into \eqref{app_oe}, and the secrecy throughput can be optimized similarly as described in the above two subsections.
	
	By ignoring the receiver noise at Eve, i.e., considering Eve's transmit power $P_e\rightarrow\infty$, one can obtain $\mathcal{F}_{\gamma_{e}}(x)=1-A_{1}(x)$. Furthermore, if Eve has more antennas than Alice, i.e., $M_e\geq M$, one have $A_1(x)=1$, $\mathcal{F}_{\gamma_{e}}(x)=0$, and accordingly $\mathcal{O}_e=1$. 
	This means, when $P_e\rightarrow\infty$, Eve with enough antennas can completely eliminate all the AN signals with an MMSE receiver such that her SNR will approach infinity. As a consequence, the SOP constraint can no longer be satisfied for any chosen rate redundancy, and no positive secrecy rate can be achieved from the perspective of secrecy outage. In other words, the null-space AN scheme can safeguard secure transmissions well for the finite blocklength regime only when the eavesdropper has fewer antennas than the transmitter, and this conclusion is the same as that for the infinite blocklength case.
}

\section{Conclusions}
This paper investigated the design of secure transmissions in slow fading channels, where secrecy encoding with finite blocklength  was employed to confront the eavesdropper. 
Both adaptive and non-adaptive schemes were devised to maximize the secrecy throughput, providing the optimal threshold of the on-off transmission policy, blocklength, code rates, and power allocation of the AN scheme.
{ Theoretical and numerical results showed that, under the on-off policy, increasing the blocklength can simultaneously enhance the reliability and secrecy, and thus the secrecy throughput is maximized when using the maximal blocklength. In addition, since an overly large secrecy rate will significantly decrease the successful decoding probability thus lowering the secrecy throughput, there exists a critical secrecy rate, but not as large as possible, that can achieve the maximal secrecy throughput. }

\appendix
\subsection{Proof of Lemma \ref{lemma_opt_re_n}}
\label{appendix_lemma_opt_re_n}
For tractability, a {piece-wise} linear approximation approach is leveraged to approximate the $Q$-function given in \eqref{error_pro}, i.e., $Q\left(\frac{C_i-R_i}{\sqrt{V_i/n}}\right)\approx \Xi(\gamma_i,n,R_i)$ for $i\in\{b,e\}$ \cite{Makki2014Finite,Makki2016Wireless},\footnote{{ This approximation has been extensively applied to the finite-blocklength scenarios, and its accuracy has been well validated.}} with 
\begin{align}\label{app_q}
\Xi(\gamma_i,n,R_i)
=\begin{cases}
0, & \gamma_i>\tau^u_i,\\
\frac{1}{2}-\frac{\beta}{\theta_i}(\gamma_i - \theta_i^2), & \tau^l_i\le\gamma_i\le\tau^u_i,\\
1,& \gamma_i<\tau^l_i,\\
\end{cases}
\end{align}
where $\beta\triangleq \frac{\sqrt{n}}{2\pi}$, $\theta_i\triangleq\sqrt{2^{R_i}-1}$, $\tau^u_i\triangleq\theta_i^2+\frac{\theta_i}{2\beta}$, and  $\tau^l_i\triangleq\theta_i^2-\frac{\theta_i}{2\beta}$.\footnote{Generally, $\theta_i>\frac{1}{2\beta}$ or $R_i>\log_2\left(1+{\pi^2}/{n}\right)$ should be satisfied to ensure a positive $\tau_i^l$.}
With \eqref{app_q}, the information leakage probability $\mathcal{O}_e$ defined in \eqref{oe} is calculated as 
\begin{align}\label{app_oe}
\mathcal{O}_e =1 - \mathbb{E}_{\gamma_e}\left[\Xi(\gamma_e,n,R_e)\right]= 1 -\frac{\beta}{\theta_e} \int_{\tau^l_e}^{\tau^u_e}\mathcal{F}_{\gamma_e}(\gamma)d\gamma,
\end{align}
where $\mathcal{F}_{\gamma_i}(\gamma) = 1 - e^{-\gamma/\Gamma_i^2}$ is the CDF of $\gamma_i$ for $i\in\{b,e\}$, and the last equality in \eqref{app_oe} follows from invoking \eqref{app_q} and using partial integration.
Next, treat $n$ as a continuous variable. As $R_e^*$ satisfies $\mathcal{O}_e(R_e^*)=\delta$, the derivative $\frac{dR_e^*}{dn}$ is obtained by using the derivative rule for implicit functions \cite{Zheng2015Multi} with $\mathcal{O}_e(R_e^*)=\delta$, i.e.:
\begin{equation}\label{dRe}
\frac{dR_e^*}{dn} = - \frac{\partial \mathcal{O}_e/\partial n }{\partial \mathcal{O}_e/\partial R_e^*}.
\end{equation}	
First, it can be proved that $\frac{\partial \mathcal{O}_e}{\partial R_e^*}=\frac{\partial \mathcal{O}_e}{\partial \theta_e}\frac{\partial \theta_e}{\partial R_e^*}<0$ by noting that $\frac{\partial \theta_e}{\partial R_e^*}>0$ and 
	\begin{align}
	&\frac{\partial \mathcal{O}_e}{\partial \theta_e} = \frac{\beta}{\theta_e^2}\int_{\tau^l_e}^{\tau^u_e}\mathcal{F}_{\gamma_e}(\gamma)d\gamma-\frac{\beta}{\theta_e}\left[\frac{d\tau^u_e}{d\theta_e}\mathcal{F}_{\gamma_e}(\tau^u_e)-\frac{d\tau^l_e}{d\theta_e}\mathcal{F}_{\gamma_e}(\tau^l_e)\right]\nonumber\\
	& \stackrel{\mathrm{(a)}}\le \frac{\beta}{\theta_e^2}\left[\gamma \mathcal{F}_{\gamma_e}(\gamma)\right]|^{\tau^u_e}_{\tau^l_e}-\frac{4\beta\theta_e+1}{2\theta_e}\mathcal{F}_{\gamma_e}(\tau^u_e)+\frac{4\beta\theta_e-1}{2\theta_e}\mathcal{F}_{\gamma_e}(\tau^l_e) \nonumber\\
	& = \beta \left[\mathcal{F}_{\gamma_e}(\tau^l_e)-\mathcal{F}_{\gamma_e}(\tau^u_e)\right]<0,\nonumber
	\end{align}
where $\mathrm{(a)}$ follows from the partial integration. 
The next step is to determine the sign of $\frac{\partial \mathcal{O}_e}{\partial n}=\frac{\partial \mathcal{O}_e}{\partial \beta}\frac{\partial \beta}{\partial n}$. The first and second derivatives of $\mathcal{O}_e$ w.r.t. $\beta$ are respectively given by 
\begin{align}
\frac{\partial \mathcal{O}_e}{\partial \beta}&=\frac{1}{2\beta}\left[\mathcal{F}_{\gamma_e}(\tau^u_e)+\mathcal{F}_{\gamma_e}(\tau^l_e)\right]-\frac{1}{\theta_e}\int_{\tau^l_e}^{\tau^u_e}\mathcal{F}_{\gamma_e}(\gamma)d\gamma,\label{dqe1}\\
\frac{\partial^2 \mathcal{O}_e}{\partial \beta^2} &= \frac{\theta_e}{4\beta^3}\left[f_{\gamma_e}(\tau^l_e)-f_{\gamma_e}(\tau^u_e)\right].\label{dqe2}
\end{align}
It is easy to see $\frac{\partial^2 \mathcal{O}_e}{\partial \beta^2}>0$ as $f_{\gamma_e}(\gamma) = \frac{1}{\Gamma_e}e^{-{\gamma}/{\Gamma_e}}$ decreases with  $\gamma$ and $\tau^u_e>\tau^l_e$. This indicates that $\frac{\partial \mathcal{O}_e}{\partial \beta}$ increases with $\beta$ such that $\frac{\partial \mathcal{O}_e}{\partial \beta}<\frac{\partial \mathcal{O}_e}{\partial \beta}|_{\beta\rightarrow\infty}=0$. Combining $\frac{\partial \mathcal{O}_e}{\partial \beta}<0$ and $\frac{\partial \beta}{\partial n}>0$ yields $\frac{\partial \mathcal{O}_e}{\partial n}<0$. With $\frac{\partial \mathcal{O}_e}{\partial R_e^*}<0$ and $\frac{\partial \mathcal{O}_e}{\partial n}<0$ in \eqref{dRe}, $\frac{dR_e^*}{dn}<0$ is obtained, which completes the proof.

\subsection{Proof of Theorem \ref{theorem_opt_n}}
\label{appendix_theorem_opt_n}
The derivative of $p_s$ w.r.t. $n$ is given by
\begin{equation}\label{dps}
\frac{dp_s}{dn}=  \frac{1}{\sqrt{2\pi}}e^{-\frac{n(C_b-R_t)^2}{2V_b}}\left[\frac{C_b-R_t}{2\sqrt{nV_b}}-\sqrt{\frac{n}{V_b}}\frac{dR_e^*}{dn}\right],
\end{equation}
which follows from the derivative $\frac{dQ(x)}{dx} = \frac{-1}{\sqrt{2\pi}}e^{-\frac{x^2}{2}}$. 
Plugging $\frac{dR_e^*}{dn}<0$, as shown in Lemma \ref{lemma_opt_re_n}, into \eqref{dps} yields $\frac{dp_s}{dn}>0$. 
For a fixed $R_s$ in \eqref{st_max1}, it is directly concluded that $\frac{d\mathcal{T}_{\rm A}(\eta)}{dn}>0$, which means that $\mathcal{T}_{\rm A}(\eta)$ monotonically increases with $n$.  {Since $n$ is an integer, $\mathcal{T}_{\rm A}(\eta)$ is maximized at the maximal integer of $n$, i.e., $n = N$, which completes the proof.}

\subsection{Proof of Theorem \ref{theorem_opt_rs}}
\label{appendix_theorem_opt_rs}
	From  \eqref{dTs}, it is easy to prove that $\frac{d^2\mathcal{T}_{\rm A}(\eta)}{dR_s^2}<0$, i.e., $\mathcal{T}_{\rm A}(\eta)$ is concave on $R_s$. 
	It is verified that $\frac{d\mathcal{T}_{\rm A}(\eta)}{dR_s}|_{R_s = 0}=1 - Q\left(\frac{C_b-R_e^*}{\sqrt{V_b/N}}\right)>0$. As a result, $\mathcal{T}_{\rm A}(\eta)$ is maximized at the boundary $R_s = C_b- R_e^*$ if $\frac{d\mathcal{T}_{\rm A}(\eta)}{dR_s}|_{R_s = C_b- R_e^*}=\frac{1}{2}-\frac{C_b-R_e^*}{\sqrt{2\pi V_b/N}}\ge 0 \Rightarrow R_e^*\ge C_b-\sqrt{\frac{\pi V_b}{2N}}$ or otherwise at the unique zero-crossing of $\frac{d\mathcal{T}_{\rm A}(\eta)}{dR_s}$, i.e.,  $R_s^{\circ}$.
	Next, the condition $R_e^*\ge C_b-\sqrt{\frac{\pi V_b}{2N}}$ is equivalently transformed to that $\gamma_b$ does not exceed a critical value $\gamma_b^{\circ}$. Let $\psi(\gamma_b)=C_b-\sqrt{\frac{\pi V_b}{2N}}$. It can be readily confirmed that $\psi(\gamma_b)<0$, and $\psi(\gamma_b)$ decreases with $\gamma_b$ if $0<\gamma_b<\gamma_b^{L}\triangleq\sqrt{\frac{1}{2}+\sqrt{\frac{1}{4}+\frac{\pi}{2N}}}-1$ or otherwise increases with $\gamma_b$. This leads to $R_e^*\ge\psi(\gamma_b)\Rightarrow\gamma_b\le\gamma_b^{\circ}\triangleq \psi^{-1}(R_e^*)$. An upper bound for $\gamma_b^{\circ}$ is further provided by realizing that $\psi(\gamma_b^{\circ})=R_e^*>\log_2(1+\gamma_b^{\circ})-\sqrt{\frac{\pi }{2N}}\log_2e\Rightarrow \gamma_b^{\circ}<\gamma_b^{U}\triangleq e^{\sqrt{\frac{\pi}{2N}}+R_e^*\ln 2}-1$. Then, $\gamma_b^{\circ}$ can be quickly calculated using the bisection method with $\psi(\gamma_b)=R_e^*$ in the range $(\gamma_b^{L},\gamma_b^{U})$. This completes the proof.
	
\subsection{Proof of Theorem \ref{theorem_opt_rs_non}}\label{appendix_theorem_opt_rs_non}
First, display the derivative $\frac{dY(\theta_b)}{d\theta_b}$ in a recursive form  $\frac{dY(\theta_b)}{d\theta_b}=\frac{1}{\theta_b}-\left(\frac{1}{\theta_b}+\frac{1}{2\beta\Gamma_b}\right)Y(\theta_b)$ with $Y(\theta_b)$ in \eqref{st_theta}. Then, the derivative $\frac{d\mathcal{T}_{\rm N}}{d\theta_b}$ is given by $\frac{d\mathcal{T}_{\rm N}}{d\theta_b} =\frac{\theta_b}{1+\theta_b^2 }e^{-\frac{\theta_b^2}{\Gamma_b}}G(\theta_b)$,
with $G(\theta_b)$ presented in \eqref{G}. It is easily proved that $G(\theta_b)>0$ when $\theta_b = \sqrt{2^{R_e^*}-1}$ and $G(\theta_b)<0$ as $\theta_b\rightarrow\infty$.
The key step of the proof is to argue that $G(\theta_b)$ monotonically decreases with $\theta_b$, which guarantees a unique zero-crossing of $G(\theta_b)$ within  $\theta_b\in(\sqrt{2^{R_e^*}-1},\infty)$. In other words, ${\mathcal{T}_{\rm N}}$ initially increases and then decreases with $\theta_b$ and reaches the maximum when $\theta_b$ arrives at the unique zero-crossing of $G(\theta_b)$. To this end, it is necessary to calculate the derivative $\frac{dG(\theta_b)}{d\theta_b}${:}
\begin{equation}\label{dG}
\frac{dG(\theta_b)}{d\theta_b} = \frac{\frac{dY(\theta_b)}{d\theta_b}-2g(\theta_b)}{\ln2}-\left[\log_2(1+\theta_b^2)-R_e^*\right]h(\theta_b),
\end{equation}
where $h(\theta_b)=\left(1-\frac{1}{\theta_b^2}\right)g(\theta_b)+\frac{1+\theta_b^2}{\theta_b}\frac{dg(\theta_b)}{d\theta_b}$. 
To {proceed}, the following lemma is {introduced}.
\begin{lemma}\label{W}
	$Y(\theta_b)$ decreases with $\theta_b$ and satisfies
	\begin{equation}\label{bound}
	\frac{2\beta\Gamma_b}{\theta_b+2\beta\Gamma_b}<Y(\theta_b)<\min\left\{1,\frac{2\beta\Gamma_b}{\theta_b}\right\}.
	\end{equation}
	\end{lemma}
	\begin{proof}
		Define $x\triangleq \frac{\theta_b}{2\beta\Gamma_b}$ such that $Y(\theta_b) = \frac{1-e^{-x}}{x}$.
		The monotonicity of $Y(\theta_b)$ w.r.t. $\theta_b$ is due to $\frac{d Y(\theta_b)}{d\theta_b}=\frac{(1+x)e^{-x}-1}{x}<0$.
		The lower bound of $Y(\theta_b)$ is obtained from $\frac{dY(\theta_b)}{d\theta_b}=\frac{1}{\theta_b}-\left(\frac{1}{\theta_b}+\frac{1}{2\beta\Gamma_b}\right)Y(\theta_b)<0$ and the upper bound follows from $Y(\theta_b) <\frac{1}{x}$ and $1-e^{-x}<x$.
		\end{proof}
		
		With the lower bound of $Y(\theta_b)$ given in \eqref{bound}, it can be readily proved that $g(\theta_b)>0$ such that the term ${\frac{dY(\theta_b)}{d\theta_b}-2g(\theta_b)}$ in \eqref{dG} is negative. Besides, since $h(\theta_b)\ge0$ directly yields $\frac{dG(\theta_b)}{d\theta_b}<0$, one only needs to discuss the situation $h(\theta_b)<0$ and prove that 
\begin{align}\label{dG1}
\frac{dG(\theta_b)}{d\theta_b} \ln 2&\le {\frac{dY(\theta_b)}{d\theta_b}-2g(\theta_b)}-h(\theta_b)\ln(1+\theta_b^2)\nonumber\\
&\stackrel{\mathrm{(a)}}<{\frac{dY(\theta_b)}{d\theta_b}-2g(\theta_b)}-\theta_b^2h(\theta_b)\nonumber\\
&\stackrel{\mathrm{(b)}}\le-\frac{1}{\theta_b+2\beta\Gamma_b}\left(\frac{\theta_b^2}{\Gamma_b}+\frac{\theta_b^4}{\Gamma_b}+8\beta\theta_b+8\beta\theta_b^3 -\theta_b^2\right)\nonumber\\
&\stackrel{\mathrm{(c)}}<0,
\end{align}
		where $\mathrm{(a)}$ is due to $\ln(1+x)\le x$, $\mathrm{(b)}$ holds by invoking Lemma \ref{W} along with algebraic manipulations, and $\mathrm{(c)}$ derives from $8\beta\theta_b+8\beta\theta_b^3\ge16\beta\theta_b^2>\theta_b^2$ as $\beta=\frac{\sqrt{N}}{2\pi}>\frac{1}{8}$.
		
\subsection{Proof of Theorem 
\ref{theorem_opt_alpha}}
\label{appendix_theorem_opt_alpha}
First fix $R_s$, and it is clear that the term $\frac{-\lambda_bR_s}{\sqrt{\lambda_b^2-1}}$ in \eqref{ps_multi} increases with $\alpha$ as  $\lambda_b$ increases with $\alpha$. It is also verified that the term $Z(\alpha)\triangleq\frac{\lambda_b(\ln\lambda_b - \ln\lambda_e )}{\sqrt{\lambda_b^2-1}}$ in \eqref{ps_multi} increases with $\alpha$ by computing the derivative of $Z(\alpha)$ w.r.t. $\alpha$:
	\begin{align}\label{dz}
	\frac{dZ(\alpha)}{d\alpha} &= 
	\frac{\frac{d\lambda_b}{d\alpha}\left(\lambda_b^2-1-\ln\frac{\lambda_b}{\lambda_e}\right)-\frac{\lambda_b(\lambda_b^2-1)}{\lambda_e}\frac{d\lambda_e}{d\alpha}}{(\lambda_b^2-1)^{3/2}}\nonumber\\
	&\stackrel{\mathrm{(a)}}=\frac{\lambda_b(\lambda_b-1)\left[(\lambda_b-\lambda_e)(\lambda_b+1)-\ln\frac{\lambda_b}{\lambda_e}\right]}{\alpha(\lambda_b^2-1)^{3/2}}\nonumber\\
	&\stackrel{\mathrm{(b)}}\ge\frac{\lambda_b(\lambda_b-1)(\lambda_b-\lambda_e)\left(\lambda_b+1-\frac{1}{\lambda_e}\right)}{\alpha(\lambda_b^2-1)^{3/2}}\stackrel{\mathrm{(c)}}>0,
	\end{align}
{ where $\mathrm{(a)}$ holds by substituting $\frac{d\lambda_i}{d\alpha} =\frac{\lambda_i(\lambda_i-1)}{\alpha}$ for $i\in\{b,e\}$, $\mathrm{(b)}$ follows from the inequality $\ln\frac{\lambda_b}{\lambda_e}\le\frac{\lambda_b-\lambda_e}{\lambda_e}$ with $\lambda_b>\lambda_e>0$, and $\mathrm{(c)}$ is due to $\lambda_b>\lambda_e>1$.} Hence, $p_s$ in \eqref{ps_multi} increases with $\alpha$ as $Q(x)$ decreases with $x$. This indicates, $\alpha^*=1$ is optimal for maximizing $\mathcal{T}_{\rm A}(\eta)= R_sp_s$ for any given $R_s$ and $\eta$ and is also optimal for maximizing $\mathcal{T}_{\rm A}$.

\subsection{Proof of Theorem \ref{theorem_opt_phi}}
\label{appendix_theorem_opt_phi}
Let $L(\phi)=L_1L_2$, where $L_1 =\frac{\lambda_b}{\sqrt{\lambda_b^{2}-1}}$ and $L_2 = {\ln\frac{\lambda_b}{\lambda_e}-R_s\ln2}$ such that  $\frac{dL_1}{d\phi}=-\rho_b(L_1^2-1)^{3/2}$ and $\frac{dL_2}{d\phi} = \frac{\rho_b}{\lambda_b}-\frac{\rho_e+\phi\frac{d\rho_e}{d\phi}}{\lambda_e}$. Rewrite the second derivative as $\frac{d^2L(\phi)}{d\phi^2}=L_1(L_1^2-1)^2 I(\phi)$ with $I(\phi)$ given by \eqref{I} at the top of this page, { where $\mathrm{(a)}$ holds by recalling the definition $L_2 \leq \ln\frac{\lambda_b}{\lambda_e}$ and invoking the result $\frac{d^2\rho_e}{d\phi^2}>\frac{2}{\rho_e}\left(\frac{d\rho_e}{d\phi}\right)^2>0$ from Lemma \ref{dchi}, $\mathrm{(b)}$ follows from plugging $L_1 =\frac{\lambda_b}{\sqrt{\lambda_b^{2}-1}}$, using the inequality $\ln\frac{\lambda_b}{\lambda_e}\le \frac{\lambda_b-\lambda_e}{\lambda_e}$, and omitting the term $(\phi^2+\frac{2\phi}{\rho_e})(\frac{d\rho_e}{d\phi})^2$, $\mathrm{(c)}$ holds by substituting $\frac{dL_2}{d\phi}$ into $\mathrm{(b)}$, and $\mathrm{(d)}$ is established after some manipulation operations and by discarding the negative term  $\frac{2(\lambda_b^2-1)}{\lambda_b\lambda_e^2}\left[\phi\rho_b\lambda_e - \lambda_b\left(\lambda_b^2-1\right)\right]$ noting that $\lambda_b = 1+\phi\rho_b>\lambda_e$.} As indicated by \eqref{I} that  $L(\phi)$ is concave on $\phi$, $L(\phi)$ is maximized at $\phi=1$ if $\frac{dL(\phi)}{d\phi}|_{\phi=1}\ge0$ or otherwise at the unique zero-crossing of $\frac{dL(\phi)}{d\phi}$.
Besides,  $\frac{dL(\phi)}{d\phi}|_{\phi=1}\ge0$ is equivalent to $X(\rho_b)\ge0$ in \eqref{condition_phi_1}. Clearly, $A_1=\frac{(1+\Gamma_e)\rho_e(1)}{1+\rho_e(1)}<1$ must be ensured to yield a positive $X(\rho_b)$, with which it can be verified that $X(\rho_b)$ increases with $\rho_b$. As a consequence,  $\frac{dL(\phi)}{d\phi}|_{\phi=1}\ge0$ can be transformed to an explicit form with relation to $\rho_b$, namely, $\rho_b\ge\rho_b^{\circ}$. This completes the proof.
\begin{figure*}[t]
	\begin{align}\label{I}
	I(\phi)
	&=3\rho_b^2L_2-\frac{2\rho_b\frac{dL_2}{d\phi}}{L_1(L_1^2-1)^{1/2}}-\frac{1}{(L_1^2-1)^{2}}\left[{\frac{\rho_b^2}{\lambda_b^2}+\frac{2\frac{d\rho_e}{d\phi}+(\phi+\phi^2\rho_e)\frac{d^2\rho_e}{d\phi^2}-\rho_e^2-\phi^2(\frac{d\rho_e}{d\phi})^2}{\lambda_e^2}}\right]\nonumber\\ &\stackrel{\mathrm{(a)}}\le 3\rho_b^2\ln\frac{\lambda_b}{\lambda_e}-\frac{2\rho_b}{L_1\sqrt{L_1^2-1}}\frac{dL_2}{d\phi}-\frac{1}{{(L_1^2-1)^2}}\left[{\frac{\rho_b^2}{\lambda_b^2}+\frac{2\frac{d\rho_e}{d\phi}+(\phi^2+\frac{2\phi}{\rho_e})(\frac{d\rho_e}{d\phi})^2-\rho_e^2}{\lambda_e^2}}\right]\nonumber\\ &\stackrel{\mathrm{(b)}}\le \left[3\rho_b^2\frac{\lambda_b-\lambda_e}{\lambda_e}-{2\rho_b}\frac{\lambda_b^2-1}{\lambda_b}\frac{dL_2}{d\phi}-\left(\lambda_b^2-1\right)^2\left({\frac{\rho_b^2}{\lambda_b^2}+\frac{2\frac{d\rho_e}{d\phi}-\rho_e^2}{\lambda_e^2}}\right)\right]\nonumber\\ &\stackrel{\mathrm{(c)}}=
	\frac{\rho_b^2\left[3\lambda_b^2(\lambda_b-\lambda_e)+\lambda_e(1-\lambda_b^4)\right]}{\lambda_b^2\lambda_e}+\frac{2\rho_b(\lambda_b^2-1)}{\lambda_b\lambda_e}\left(\rho_e+\phi\frac{d\rho_e}{d\phi}\right)+\frac{(\lambda_b^2-1)^2}{\lambda_e^2}\left(\rho_e^2-2\frac{d\rho_e}{d\phi}\right)\nonumber\\
	& \stackrel{\mathrm{(d)}}\le -\frac{\phi\rho_b^2(\rho_b-\rho_e)}{\lambda_b^2\lambda_e^2}\left[2\phi^4\rho_b^3\rho_e+\phi^3\rho_b^2(\rho_b+6\rho_e)+\phi^2\rho_b(\rho_b+8\rho_e)+5\phi\rho_e+1\right]<0,
	\end{align}
	\hrulefill
	\end{figure*}

	\subsection{Proof of Corollary \ref{corollary_opt_phi}}
	\label{appendix_corollary_opt_phi}
	Let $K(\phi)$ denote $\frac{dL(\phi)}{d\phi}$ in \eqref{dL}. 
	It is verified that $\frac{d\phi^*}{d\eta}=-\frac{\partial K(\phi)/\partial\eta}{\partial K(\phi)/\partial\phi}|_{\phi=\phi^*}>0$ by recalling that $\frac{\partial K(\phi)}{\partial\phi}|_{\phi=\phi^*}<0$ from Theorem \ref{theorem_opt_phi} and proving that 
\begin{align}
\frac{\partial K(\phi)}{\partial\rho_b}|_{\phi=\phi^*}&=\frac{\frac{\lambda_b^2-\lambda_b+1}{\lambda_b}-(1-A_{\phi^*})+\frac{(2\lambda_b-1)(\ln \lambda_b-B_{\phi^*})}{\lambda_b+1}}{\phi^*(\lambda_b^2-1)^{{5}/{2}}/P_b}\nonumber\\
&\stackrel{\mathrm{(a)}}=\frac{\frac{1}{\lambda_b}-\lambda_b+(2\lambda_b^2-\lambda_b-1)(1-A_{\phi^*})}{\phi^*(\lambda_b^2-1)^{{5}/{2}}/P_b}\nonumber\\
&\stackrel{\mathrm{(b)}}\ge
\frac{\lambda_b-1}{\phi^*(\lambda_b^2-1)^{{5}/{2}}/P_b}>0,
\end{align}
	where $\mathrm{(a)}$ is due to $\frac{\ln \lambda_b - B_{\phi^*}}{\lambda_b+1} = {(1-A_{\phi^*}){\lambda_b}-1}$ from $K(\phi^*)=0$, and $\mathrm{(b)}$ is because $(1-A_{\phi^*}) \lambda_b-1>0$. Moreover, $\lim_{\eta\rightarrow\infty}K(\phi^*)=\frac{1-A_{\phi^*}}{\phi^*}$. Solving $K(\phi^*)=0$ with $A_{\phi^*} =\frac{\phi^*\Lambda}{(1-\phi^*)(1-\phi^*+\phi^*\Lambda)}$ yields $\phi^* = \frac{1}{\sqrt{\Lambda}+1}$.
	Similarly, one can prove that  $\frac{d\phi^*}{d\Lambda}<0\Rightarrow \frac{d\phi^*}{d\delta}>0$ and $\lim_{\delta\rightarrow 1}\phi^* = 1$. 
	
\subsection{Proof of Proposition \ref{proposition_high}}\label{appendix_proposition_high}	
		Note that $\varrho_1,\varrho_2\rightarrow 0$ as $\Gamma_b\rightarrow\infty$. Resorting to \cite[Eqn. (44)]{Zhang2013Design} yields
		\begin{equation}\label{g1}
		\bar\Gamma(M,\varrho_i) = e^{-\varrho_i}\sum_{k=0}^{M-1}\frac{\varrho_i^k}{k!} \approx 1-\frac{\varrho_i^M}{M!}, ~i\in\{1,2\},
		\end{equation}
		and the term $\Delta\Gamma$ in \eqref{st_multi} is approximated as 
		\begin{align}\label{g2}
		&\Delta\Gamma=	\sum_{k=0}^{M-1}\left[\bar\Gamma(k+1,\varrho_1)-\bar\Gamma(k+1,\varrho_2)\right] \nonumber\\
		&=\sum_{k=0}^{M-1}\left(e^{-\varrho_1}\sum_{m=0}^k\frac{\varrho_1^m}{m!}-e^{-\varrho_2}\sum_{m=0}^k\frac{\varrho_2^m}{m!}\right)\nonumber\\
		& =\sum_{k=0}^{M-1}\frac{M-k}{k!}\left(e^{-\varrho_1}{\varrho_1^k}-e^{-\varrho_2}{\varrho_2^k}\right)\nonumber\\
		&=M\Delta\Gamma(M-1,\varrho_1,\varrho_2)-\varrho_1\bar\Gamma(M-1,\varrho_1)+\varrho_2\bar\Gamma(M-1,\varrho_2)\nonumber\\
		&\approx M\left(\frac{\varrho_1^M}{M!}-\frac{\varrho_2^M}{M!}\right)-\left[\varrho_1-\frac{\varrho_1^{M}}{(M-1)!}\right]+\left[\varrho_2-\frac{\varrho_2^{M}}{(M-1)!}\right]\nonumber\\
		&=\varrho_2-\varrho_1=\frac{{\theta_b}}{2\beta\phi \Gamma_b}.
		\end{align}
		Substituting \eqref{g1} and \eqref{g2} into \eqref{st_multi} completes the proof.

\end{document}